\documentclass[%
 reprint,
 superscriptaddress, 
 amsmath,amssymb,
 aps,
 pra,
 floatfix,
]{revtex4-2}

\usepackage{graphicx}
\usepackage{dcolumn}
\usepackage{bm}
\usepackage{dsfont}
\usepackage{xcolor}
\usepackage{enumitem}
\usepackage{multirow}
\usepackage{hyperref}
\usepackage{array}
\usepackage{float}
\usepackage{csquotes}
\usepackage[T1]{fontenc}
\MakeOuterQuote{"}

\definecolor{customcolorblue}{HTML}{4573ae}

\newcommand\editcolor[1]{{#1}} 

\usepackage{tikz}
\usepackage{amsthm}
\usepackage{amssymb}

\usepackage{mathtools}

\usepackage{algpseudocode}
\usepackage{algorithm}

\newtheorem{theorem}{Theorem}

\newtheorem{lemma}[theorem]{Lemma}

\theoremstyle{definition}

\begin{document}



\title{Optimizing compilation of error correction codes for $2\times N$ quantum dot arrays and its NP-hardness}

\author{Anthony Micciche}
\email{amicciche@umass.edu}
\affiliation{Manning College of Information and Computer Sciences, University of Massachusetts Amherst, 140 Governors Drive
Amherst, Massachusetts 01003, USA}

\author{Feroz Ahmed Mian}
\affiliation{Manning College of Information and Computer Sciences, University of Massachusetts Amherst, 140 Governors Drive
Amherst, Massachusetts 01003, USA}

\author{Anasua Chatterjee}
\affiliation{
QuTech and Kavli Institute of Nanoscience, Delft University of Technology, Delft, The Netherlands}

\author{Andrew McGregor}
\affiliation{Manning College of Information and Computer Sciences, University of Massachusetts Amherst, 140 Governors Drive
Amherst, Massachusetts 01003, USA}

\author{Stefan Krastanov}
\email{skrastanov@umass.edu}
\affiliation{Manning College of Information and Computer Sciences, University of Massachusetts Amherst, 140 Governors Drive
Amherst, Massachusetts 01003, USA}

\date{\today}

\begin{abstract}
The ability to physically move qubits within a register allows the design of hardware-specific error-correction codes, which can achieve fault-tolerance while respecting other constraints. In particular, recent advancements have demonstrated the shuttling of electron and hole spin qubits through a quantum dot array with high fidelity. It is therefore timely to explore error correction architectures consisting merely of two parallel quantum dot arrays, an experimentally validated architecture compatible with classical wiring and control constraints. Upon such an architecture, we develop a suite of heuristic methods for compiling any Calderbank-Shor-Steane (CSS) error-correcting code's syndrome-extraction circuit to run with a reduced number of shuttling operations. We demonstrate how column-regular qLDPC codes can be compiled in a provably minimal number of shuttles that is exactly equal to the column weight of the code when Shor-style syndrome extraction is used. We provide tables stating the number of required shuttles for many contemporary codes of interest. In addition, we provide a proof of the NP hardness of minimizing the number of shuttle operations for general codes, even when using Shor syndrome extraction. We also discuss how one could get around this by placing blanks in the ancilla array to achieve minimal shuttles with Shor syndrome extraction on any CSS code, at the cost of longer ancilla arrays.
\end{abstract}

\maketitle


\section{Introduction}

Motivated by the potential of far easier scalability than other quantum computing architectures due to industrial fabrication compatibility\cite{zwerver_thesis} as well as long coherence times\cite{Burkard2023}, in recent years there have been numerous advancements in experimentally realizing spin qubits in electrically defined quantum dots. In particular, the repeated shuttling of an electron or hole up and down an array of quantum dots with a per-hop error rate less than 0.01\%~\cite{Wang2024,zwerver23shuttling} has been demonstrated. On this experimentally validated foundation of being able to physically move qubits along arrays of quantum dots, we explore quantum computing architectures consisting merely of two of these quantum dot arrays, as shown in Fig.~\ref{fig:hardware_abstraction}. These hardware layouts have been referred to as a "$2 \times n$" architectures \cite{ansoloni2020,ansoloni2023}. Current fabrication restrictions on the fan-out, wiring, and wirebonding of these semiconductor quantum dot arrays make it necessary to explore the possibility of fault-tolerant computation in this restricted configuration of two parallel arrays. Natural implementations of error correction codes for such architectures were recently proposed in~\cite{siegel2024}, however, without yet the capability of optimizing the syndrome extraction circuits under the significant connectivity constraints of the hardware for codes beyond the surface code. In this work, we provide techniques for efficiently compiling \emph{any} CSS quantum error correction code to such a $2\times n$ architecture. In addition, we show that "column-regular" qLDPC codes can always be compiled to run with a small constant number of shuttles, as their property guarantees the existence of a trivial compilation solution under a certain style of fault-tolerant syndrome extraction.

\begin{figure}
    \centering
    \includegraphics[width=\linewidth]{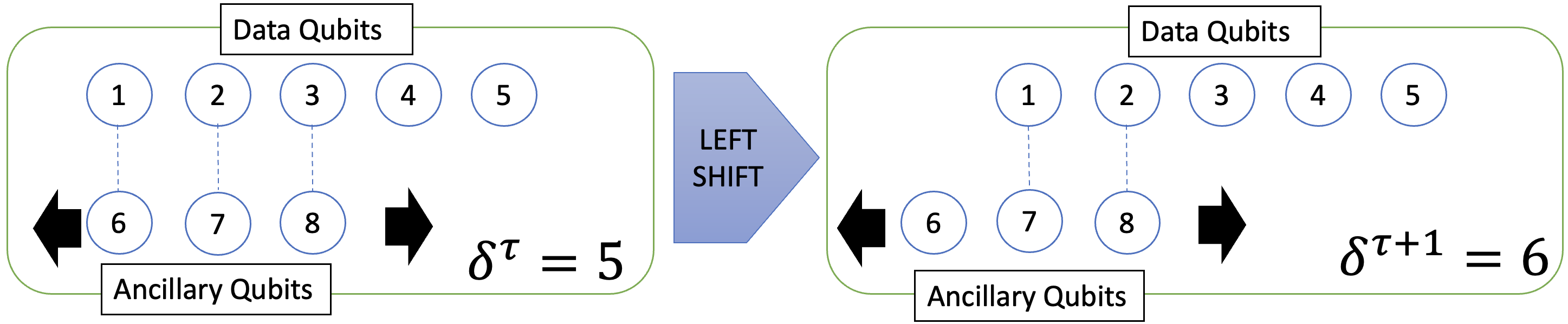}
    \caption{\textbf{Schematic of the quantum dot array hardware}. On the left is a possible configuration of this quantum memory, which we parameterize as $\delta^{\tau}=5$, where $\delta^\tau$ represents the offset between the indices of the two rows of qubits (data and ancillary qubits in the case of quantum error correction) at time $\tau$. This $\delta^\tau=5$ implies that cross-array interaction is only permitted between qubit pairs (1,6), (2,7), and (3,8). If we shuttle the bottom row to the left, we would then have the right figure which has $\delta^{\tau+1}=6$ and interaction permitted only between (1,7) and (2,8).}
    \label{fig:hardware_abstraction}
\end{figure}

More explicitly, we propose techniques that convert one quantum circuit to an equivalent circuit, which will then require fewer shuttling operations than the original to execute, thereby causing less error to accumulate. We refer to this process of modifying the initial syndrome extraction circuit as "compilation".

While portions of our approach could potentially be applied to running arbitrary circuits on this architecture, we focus our discussion and assumptions on syndrome extraction circuits used in quantum error correction, which is a natural fit for a two-row architecture as we can simply assign data qubits to the top row, and ancillary qubits to the bottom row. Shuttling these rows past each other enables us to perform gates between distant qubits, enabling the efficient implementation of codes that even require long-range connectivity such as qLDPC codes, a very promising class of codes for achieving large-scale quantum error correction~\cite{breukman2021,gottesman2014faulttolerantquantumcomputationconstant}. This work aims to provide insight on which error correction codes would be most effective for this hardware. Furthermore, we present tools for evaluating how much shuttling error an arbitrary error-correcting stabilizer code would incur. More general layouts, connecting separate $2\times n$ blocks are now also being suggested \cite{siegel2025snakesplane} and our work will be of significant benefit to such large-scale deployments as well.

This paper is structured as follows: In Section~\ref{sec:foundations}, we present a way to formalize the general minimization of qubit shuttles. In Section~\ref{sec:heuristics}, we then present a simple way of visualizing the problem formalized in the preceding section. This visualization naturally lends itself to simple heuristic methods that, while not always optimal, yield surprisingly effective results for certain circuits. Following that, in Section~\ref{sec:shor}, we discuss a special case of our problem corresponding to fault-tolerant Shor-style syndrome extraction circuits and how this special case gives us enough structure to substantially improve our compilation. After these discussions on how to reduce the number of shuttling operations, in Section~\ref{sec:comp_results} we provide many relevant examples of both column regular and irregular codes, and provide the number of shuttles our methods were able to compile their syndrome extraction circuits to. This then culminates in Section~\ref{sec:ldpc}, where we discuss how regular column weight LDPC codes compile with a provably minimal number of shuttles under fault-tolerant syndrome extraction. Further elaborating on ways to construct an error correcting scheme with minimum number of shuttle operations of the ancilla array, Section~\ref{sec:ldpc} also discusses another way of viewing the compilation problem, trading additional space in the ancilla array for the guarantee of minimum shuttling operations on any code. Finally, in the concluding remarks, we discuss some of the limitations and potential future work, such as the importance of investigating post-selected cat state preparation for use with Shor syndrome extraction.

In the Appendix, we provide a proof of the NP-hardness of compilation via reduction from 3-partition, a brief discussion of how the larger error correction pipeline would look including the cat state preparation required for fault-tolerant syndrome extraction on $2 \times n$ architecture, as well as a gate-by-gate simulation of the error correction pipeline, verifying that our techniques of $2 \times n$ compilation do indeed result in functional  syndrome circuits with better error rates corresponding to fewer shuttles. The Appendix also includes a few additional figures and examples to help illustrate the methods presented here.

\section{Foundations of $2 \times n$ Circuit Compilation\label{sec:foundations}}

We formalize the problem as follows: We are given a circuit as an array of $m$ two-qubit gates, $\mathcal{C}=A_1,\ A_2,\ ...,\ A_m$, and we say that $\delta(A)$ represents the difference in the indices of qubits involved in gate $A$. For example, if $A$ is a CNOT from a qubit at index $i$ to a qubit at index $j$, then $\delta(A)$ = $|i-j|$. We assume we can lay out our qubits such that each gate always involves one qubit from the top row and one qubit from the bottom row (e.g.\ in a syndrome extraction circuit, the top row would be data qubits and the bottom row would be ancillary qubits). \editcolor{We also let $\delta^t$ represent the state of our system at a given time step $t$, as enumerated and shown in Fig.~\ref{fig:hardware_abstraction}. That is, $\delta^t$ denotes the difference in indices between cross-row pairs of qubits. Because $\delta^t$ is thus a discrete measure of the difference of positions of the two rows and we consider 1-dimensional dot arrays, all values of delta are reachable by shuttling only one of the two rows. Therefore, as the way we define $\delta^t$ depends only on the positions of the two rows relative to each other, our abstraction holds when only one of the two rows is allowed to move, and so we assume a stationary row of data qubits throughout this work.} A more detailed sketch of our envisioned hardware is presented in Fig. \ref{fig:full_2xn}. 

Notice that a gate $A$ can execute only when $\delta(A) = \delta^t$, i.e.\ when the two rows are aligned such that the two qubits involved in gate $A$ are facing each other. It will often be convenient to track the $\delta$ values of the gates in the circuit $\mathcal{C}$, which we denote by $\mathcal{G} \vcentcolon = [\delta(A_1), \delta(A_2), ..., \delta(A_m)]$. We now discuss the two general operations that we can perform on a given syndrome-extraction circuit $\mathcal{C}$ in order to compile it into a $\mathcal{C'}$ that extracts the same syndrome, only now requiring fewer shuttle operations:

\begin{itemize}
    \item Change the order of gates in $\mathcal{C}$. We refer to this as "gate shuffling".
    \item Apply a permutation to the indices of the qubits. In other words, define a bijection from $\{1, ...,n+s\}$ to $\{1,...,n+s\}$ where $n$ is the number of data qubits and $s$ is the number of ancillary qubits. Then for each gate $A \in \mathcal{C}$, apply this bijection to both indices. We refer to this as "qubit re-indexing". For an example, see Fig.~\ref{fig:reindexing_example}.
\end{itemize}

For qubit re-indexing, there are no restrictions on how to define the bijection, as qubit indices are arbitrary in general, while in the case of gate shuffling, gates cannot arbitrarily be reordered. However, within a syndrome extracting circuit for a CSS code, all gates measuring the $X$ stabilizers may be permuted arbitrarily, and the same holds true for gates measuring $Z$ stabilizers. Similar to the idea of performing syndrome extraction on one type of stabilizers while shuttling in one direction and then extract the other stabilizers while shuttling back in the other as proposed in~\cite{siegel2024}, we also will assume our code is CSS, and compile the syndrome extraction circuits for $X$ and $Z$ checks separately (e.g.\ by re-indexing only ancillary qubits and keeping the data qubit indices fixed).

\begin{figure}
    \centering
    \includegraphics[width=\linewidth]{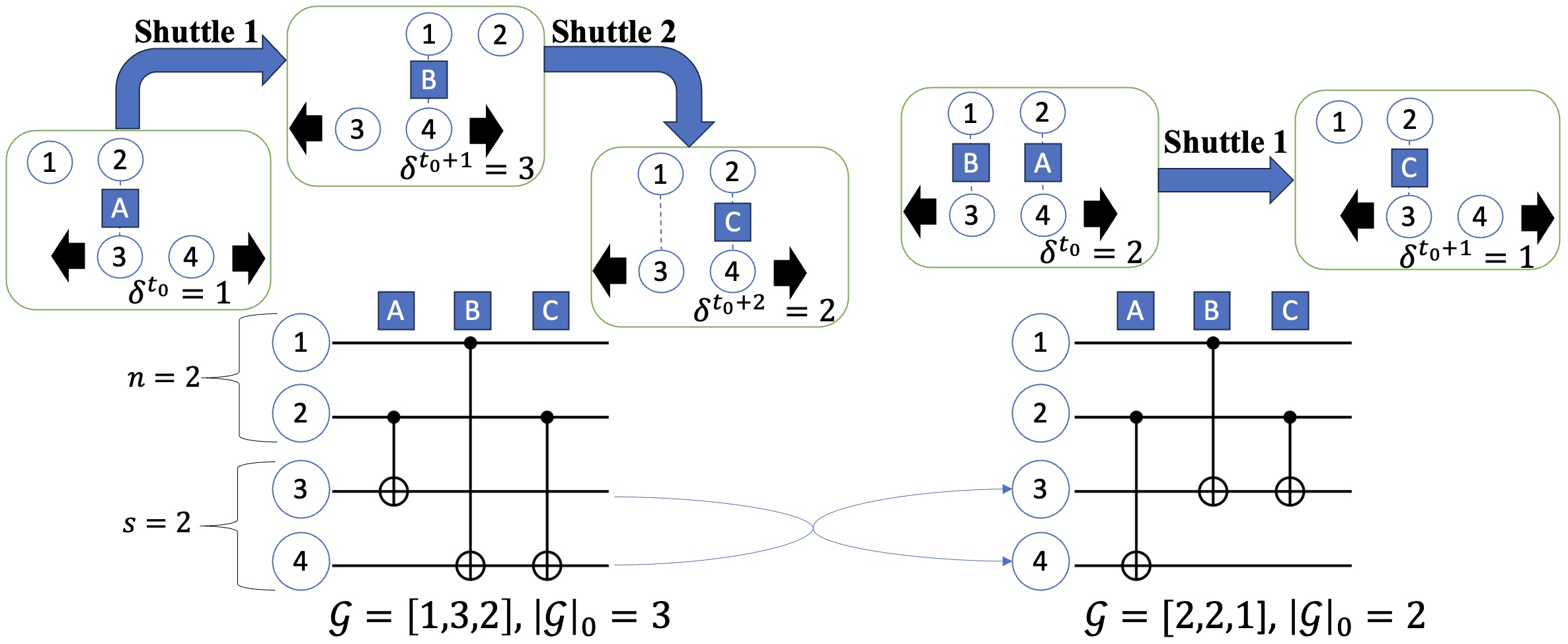}
    \caption{\textbf{Example of how qubit re-indexing reduces shuttles}. On the left is an unoptimized circuit. $\mathcal{G}$ corresponds to the values of $\delta$ for each gate in the circuit, $\delta(A), \delta(B),$ and $\delta(C)$. $|\mathcal{G}|_0$ denotes the number of unique elements in $\mathcal{G}$. Notice that to run this circuit on a $2 \times n$ architecture, two shuttles (or $|\mathcal{G}|_0$ positions) are necessary. On the right is the same circuit, except the indexing of the qubits has changed, which then allows the same circuit as before to run, this time only requiring a single shuttle.}
    \label{fig:reindexing_example}
\end{figure}

\begin{figure}
    \centering
    \includegraphics[width=\linewidth]{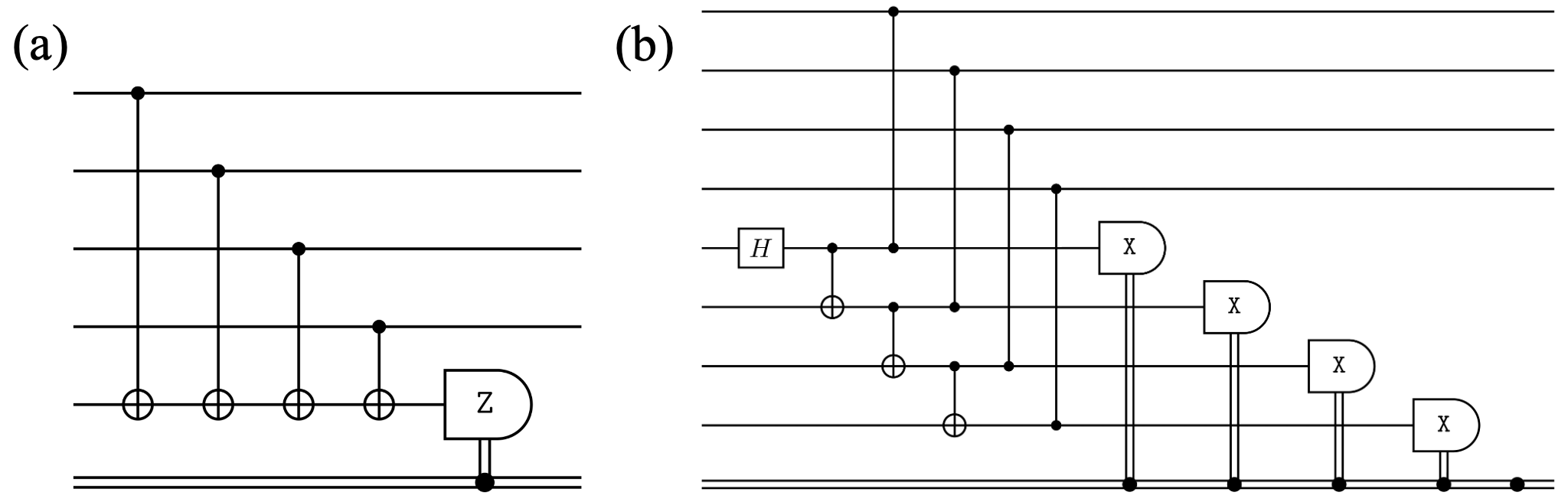}
    \caption{\textbf{Measuring the parity of $Z_1Z_2Z_3Z_4$ via two types of syndrome extraction.} a) \textbf{Naive syndrome extraction}: Parity measurement is given by coupling all pertinent data qubits to single ancillary qubit, and then performing a measurement. b) \textbf{Shor-style syndrome extraction~\cite{shor1997faulttolerant}}: Each data qubit indicated by a non-identity element in a weight $w$ stabilizer is coupled to its own ancillary qubit. For that single parity check, the participating ancilla are beforehand prepared in a $w$-qubit GHZ state. Each ancilla is then measured, and the classical XOR is taken over the resulting values to yield a the parity of the stabilizer measurement. Double lines here represent the flow of classical information.}
    \label{fig:syndrom_circuits}
\end{figure}

We refer to a circuit that measures the parity of a weight $w$ stabilizer by applying $w$ gates to a single ancillary qubit in a non-fault-tolerant way as "naive syndrome measurement" (see Fig.~\ref{fig:syndrom_circuits}a). For the purposes of compilation, this corresponds to the general problem, as each ancilla may have multiple gates attached to it. We consider only the two-qubit gates in the circuits to be compiled (i.e., single-qubit gates are folded into the neighboring two-qubit gates and not explicitly tracked). 

\subsection{Gate Shuffling\label{sec:gate_shufle}}
Before considering the more interesting problem of qubit re-indexing, we first describe how for a given ordering of the qubits, changing the gate order can reduce the number of required shuttles. This will be an implicit primitive behind the definition of the objective over which we want to minimize when solving the qubit re-indexing problem.

Starting from some arbitrary order of a circuit $\mathcal{C}=[A_1,A_2,...,A_m]$, referring to Fig.~\ref{fig:hardware_abstraction}, it should become obvious that if two distinct (acting on different qubits) gates  $A_i$ and $A_j$ have the same $\delta$ value, then both gates can run concurrently at time step $t$ if the relative position of the two rows is $\delta^t = \delta(A_i)=\delta(A_j)$. 

Now because gates with the same $\delta$ can run concurrently, by merely permuting the gates of $\mathcal{C}$ so that gates of the same $\delta$ are adjacent to one another (sorting $\mathcal{C}$ by $\delta$ is a straightforward way to achieve this), we can partition $\mathcal{C}$ into a collection of sub-arrays for each corresponding different value of $\delta$. We call the number of such sub-arrays $|\mathcal{G}|_0$, where $|\mathcal{G}|_0$ is defined as the number of unique values within $\mathcal{G}$ -- the array that contains the $\delta$ values for the current order of gates. All gates in a given sub-array can be executed concurrently, and a single shuttling operation is required to set the architecture in the correct configuration for the corresponding sub-array to be executed. We can therefore state that $|\mathcal{G}|_0$ represents the minimum number of required shuttles obtainable by gate shuffling, as we can always trivially sort the gates by their corresponding $\delta$, assuming all gates in $\mathcal{C}$ commute. 

Building on this result of how to perform gate shuffling to achieve $|\mathcal{G}|_0$ shuttles, we now discuss qubit re-indexing with the goal of finding the bijection of qubit indices that would best minimize $|\mathcal{G}|_0$, further reducing the required number of shuttle operations. To solve this, we fix the data qubits, and only think about re-indexing the ancillary qubits. While the main motivation for fixing the data qubits is to make the problem more tractable, as mentioned before, fixing the data qubits also causes the case of having different $X$ and $Z$ extraction circuits for CSS codes more manageable as well. 

\subsection{Reducing Shuttles via Ancillary Qubit Re-indexing\label{sec:qubit_reindexing}}

 For a given naive syndrome-extraction circuit $\mathcal{C}$, we would like to find a bijection of qubit indices that minimizes the resulting $|\mathcal{G}|_0$ (i.e.\ the number of necessary shuttles). As an example, we refer to Fig.~\ref{fig:reindexing_example2}, where we would like to reorder qubits 6, 7, and 8, noting that when the index of an ancillary qubit is changed, the $\delta$ values of all gates connected to it will also change. To account for this, we partition $\mathcal{G}$ into $s$ sets, where $s$ is the number of ancillary qubits, and each set will contain the $\delta$ values for the gates touching the corresponding ancillary qubit. In the example present in the left-hand portion of Fig.~\ref{fig:reindexing_example2}, $\mathcal{G} = [3,4,5,2,3,5]$ is partitioned into sets $R_6,R_7,R_8$, equal to respectively $\{3,4,5\}, \{2,3\}, \{5\}$, where $\{3,4,5\}$ are the $\delta$-values for all gates connected to qubit 6, $\{2,3\}$ are the $\delta$ values for all gates connected to qubit 7, and $\{5\}$ is the $\delta$ of the one gate connected to qubit 8. As we are interested in reordering these ancillary qubits, we need to see how their current position affected the values of $\mathcal{G}$. For this reason, we label the ancillary qubits $+1,+2,...,+s$, and notice that each position augments the values in corresponding set element-wise. For example, qubit 6 corresponds to the set $R_6 =\{3,4,5\}$, which we can then write as $R_6 = \{2,3,4\}+1$, where the $+1$ encodes that this is the first ancillary qubit. We refer to the set $P_i$ of $R_i = P_i+k_i$ as the \textit{primitive set} of the $i^{\text{th}}$ ancillary qubit ($k_i$ is the integer corresponding to the current position of $P_i$). The primitive set encodes the relative positions between all top-row (data) qubits that interact with a given bottom-row (ancillary qubit), while $k_i$ encodes the current position of that ancillary qubit in the bottom row. On the right-hand side of Fig.~\ref{fig:reindexing_example2}, the ancillary qubits have been reordered to minimize $|\mathcal{G}|_0$. Notice how the primitive sets remain the same under this reordering (how to derive the primitive sets from a given error correcting code is discussed in more detail in Section~\ref{sec:ldpc}). Furthermore, we can concretely define an admittedly not very friendly combinatorial optimization problem as:\\

\textbf{Problem Formulation}: Given a collection of $s$ sets $\mathcal{S} = \{P_1,P_2,...,P_s\}$, find a bijection $\pi: [s] \rightarrow [s]$, such that $|\mathcal{G}|_0 $ is minimized ($\mathcal{G}$ depends on the $R$ sets which each depend on their corresponding $P$ set and the index reordering $\pi$). More formally: 
\begin{equation}
    \min_\pi |\mathcal{G}|_0  = \min_\pi \Bigl| \bigcup_{i =1}^s \bigcup_{j \in P_{i}} \{j+ \pi(i) \}\Bigr|.
\end{equation}

As $|\mathcal{G}|_0$ represents the number of required shuttles, this procedure of minimizing $|\mathcal{G}|_0$ by reordering the ancillary qubits results in a circuit that is more efficient by requiring fewer shuttles and would then have a lower execution time due to increased gate parallelism. 

Although we prove in Appendix~\ref{sec:NP_proof} that, even in the special case where all primitive sets are cardinality $1$, finding the optimal reidexing is an NP-hard problem, in the following section, we provide some simple heuristics to act as a baseline for solving this problem. Furthermore, in section~\ref{sec:shor}, we provide a more substantial heuristic algorithm for minimizing $|\mathcal{G}|_0$ when we are promised that the input circuit is in the style of Shor syndrome extraction (primitive sets are cardinality $1$).

\begin{figure}
    \centering
    \includegraphics[width=\linewidth]{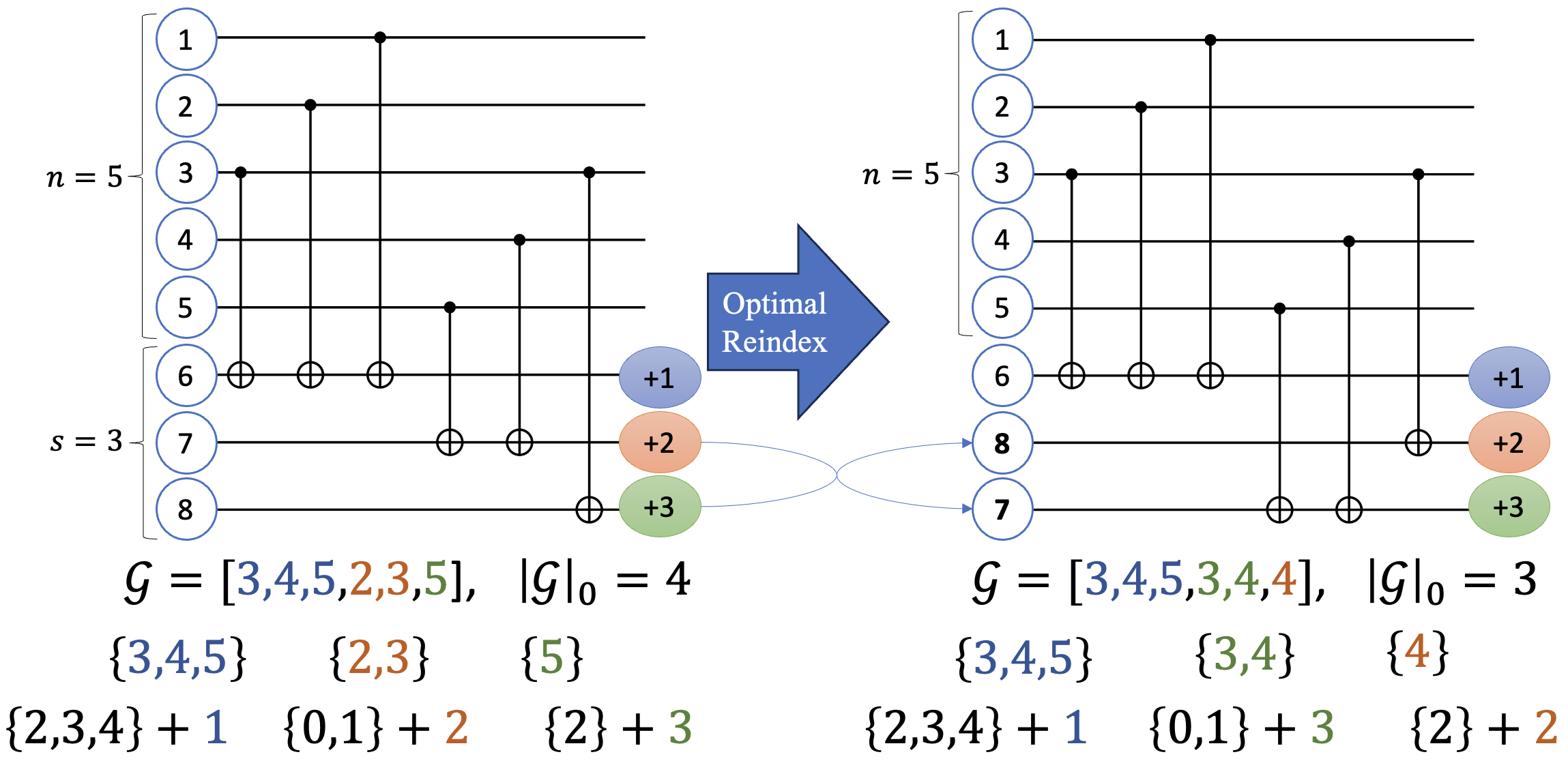}
    \caption{\textbf{Illustration of primitive sets under different re-orderings.} On the left is an arbitrary circuit, and below are the $\delta$ values for each of its gates, $\mathcal{G}$, as well as the number of unique values, $|\mathcal{G}|_0$. Below that, we partition $\mathcal{G}$  into three sets, each corresponding to the gates connected to each ancillary qubit. Below each set is a "primitive set" with some integer being added element-wise. This addition encodes the current indexing of the ancillary qubits. On the right, is a different way of indexing these ancillary qubits. Notice that although the values in $\mathcal{G}$ have changed, the primitive sets remain the same.
}
    \label{fig:reindexing_example2}
\end{figure}

\section{Heuristics for General Compilation\label{sec:heuristics}}
In this section, we discuss how to develop heuristics for the naive syndrome qubit re-indexing problem stated in the previous section. The first step of which is converting the problem to something that is easier to visualize. Again, referring to the example present in Fig.~\ref{fig:reindexing_example2}, we now recast it to stacking blocks on a "staircase". We do so by thinking of each primitive set as a chain of blocks, of which there are two varieties and the staircase to discuss (Fig.~\ref{fig:blocks_example_1} for a visual reference):
\begin{itemize}
    \item \textbf{Solid blocks} (Blue) These blocks represent the presence of an element, or what used to be a target ($\oplus$ for CNOT gates) in the quantum circuit diagram. As we are fixing the data qubits, we can also think of solid blocks as gates.
    \item \textbf{Air blocks} (Green) These blocks represent the absence of an element and therefore are referred to as "air". The reason for them is to preserve the structure and constraints of the original problem.
    \item \textbf{Staircase} (Gray) The staircase represents the set $A= \{1,...,s\}$ that we are assigning primitive sets to. By placing $P_i$ on the first step of the staircase, we thereby are shifting all elements of $P_i$ to the right by one, visually adding 1 to the set element-wise. If we were to place $P_i$ on the $j$th step, then we would be adding $j$ to $P_i$ element-wise, and the resulting set corresponds to what we call $R_i$. The goal then is to place each chain of blocks on this staircase such that the number of columns containing at least one blue block is minimized, as each column in this staircase abstraction corresponds to a particular configuration of the $2 \times n$ architecture. Moreover, minimizing the number of required configurations of the two rows of qubits is exactly the same as minimizing the number of required shuttling operations. 
\end{itemize}

In order to create the corresponding chain of blocks for a given primitive set $P$, we start with a number of blocks equal to the largest element in $P$ plus one. Then enumerating those blocks from $0$ to the largest element in $P$, blocks that correspond to elements in $P$ are solid blocks, and all other blocks are air blocks. The example provided in Fig.~\ref{fig:blocks_example_1} has only contiguous solid blocks, however, in general this will not be the case and is a source of difficulty in finding good heuristics (Fig.~\ref{fig:steane_staircase} has an example of non-contiguous solid blocks). The main benefit of this stacking block interpretation is that now it becomes much more simple to reason about heuristics and the problem as a whole. For instance, two rather decent initial heuristics for finding good orders are:
\begin{itemize}
    \item Sort these chains of blocks by total length, placing the longest ones at the top (+1 is the top of the staircase)
    \item Sort these chains by the number of air blocks preceding the first solid block
\end{itemize}
 Because each of these will only take $O(n\log(n))$ time, it is quite efficient to run both and simply take the better of the two choices, breaking ties with the heuristic not being primarily sorted by. Thereupon, as the map from the old order to the new order discovered for packing blocks on the staircase directly corresponds to the bijection $\pi$ mentioned in the problem formulation, we can apply this mapping to the ancillary qubits in the original circuit to a create a new compiled circuit. We call this method ancillary heuristic re-indexing (AHR). Notice that both sorting metrics would work in the simple problem of Fig.~\ref{fig:blocks_example_1}. To see these heuristics finding optimal solutions on real (although pedagogical) codes, such as Shor's 9 qubit code, and Steane's 7 qubit code, refer to Appendix~\ref{sec:more_block_examples}. These techniques also are sufficient to find optimal ancilla orders on column-regular LDPC codes when using Shor syndrome extraction (see Table.~\ref{tab:shor_XZ}). A third heuristic which is also checked when we mention AHR (the best of three is chosen as the output) is to sort all of the primitive sets by the first element in their set in descending order. Values that have occurred already are de-prioritized according to the number of times they've occurred. For example, this would sort $\{6,6,6,5,5,4,3,3,2,1,1,0\}$ as $\{6,5,4,3,2,1,0,6,5,3,1,6\}$. This was originally introduced as first guess to solving the case when all primitive sets are size 1 (this corresponds to Shor-syndrome extraction and is discussed in the following section), however, we find that occasionally this method finds a decent solution when the first two heuristics fail. For example, this third heuristic is the best of the three for the toric code's naive syndrome extraction circuit. 

\begin{figure}
    \centering
    \includegraphics[width=\linewidth]{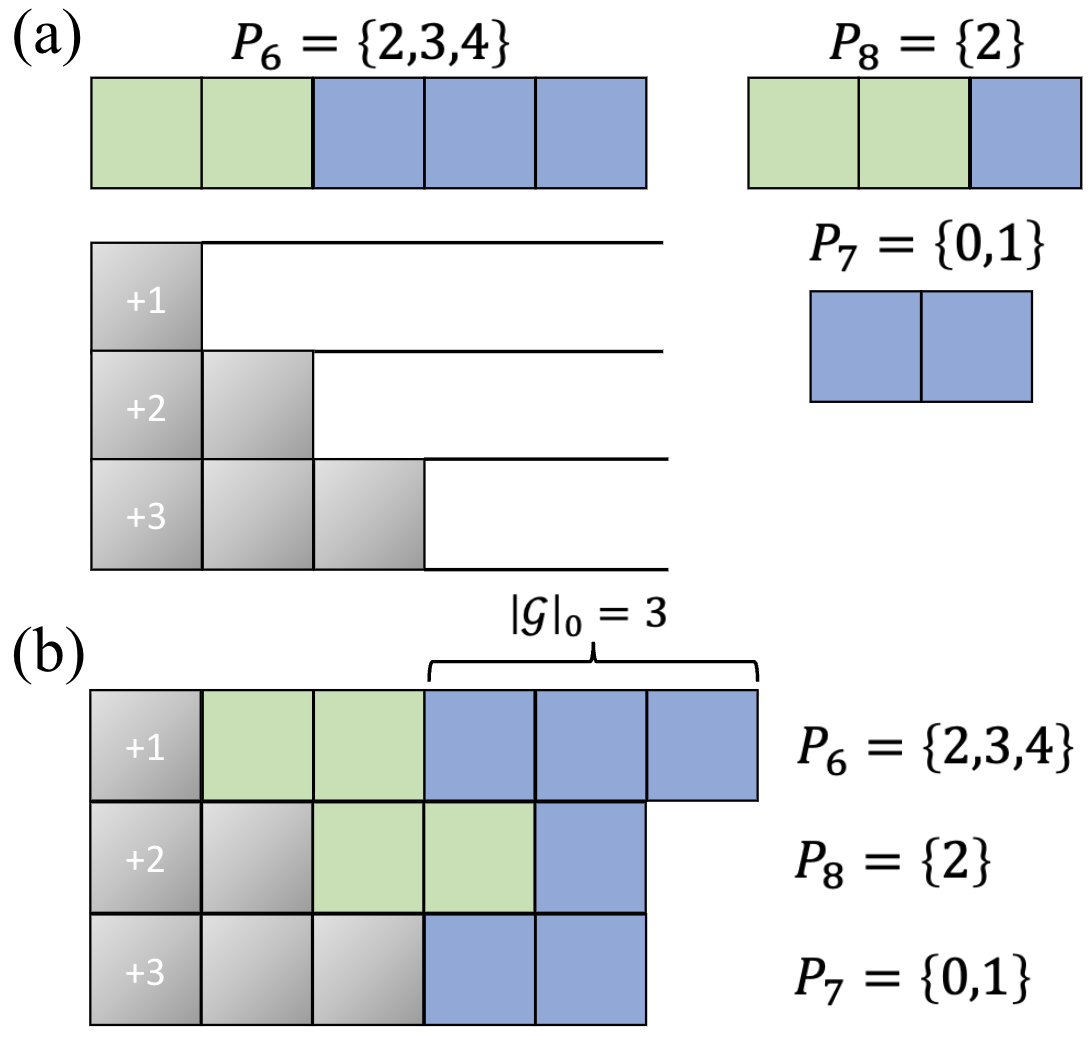}
    \caption{\textbf{Visualization of block placement}. a) Primitive sets of the circuit in Fig.~\ref{fig:reindexing_example2} represented in our so-called "block form". The green blocks represent the absence of an element in the set, while the blue blocks represent the presence of the corresponding element. Placing these blocks on the staircase represent assignments. For examples, placing $P_6$ on +1 would cause all its blocks to shift to the right by one. After placing all blocks, all columns that contain at least one blue block would correspond to a unique element in $\mathcal{G}$, and recall unique values in $\mathcal{G}$ correspond to different configurations of the $2 \times n$ architecture and therefore require shuttles. b) Optimal placement for this problem.
}
    \label{fig:blocks_example_1}
\end{figure}

\section{Shor-style fault-tolerant Circuit Compilation\label{sec:shor}}
In the previous section we discussed the general problem of assigning primitive sets to integers. However, the circuits that the sets were derived from were not quite representative of fault-tolerant quantum error correcting circuits, as they do not provide any protection against the propagation of errors to the ancillary qubits. To remedy this and to assess a simplified version of the problem, in this section we consider a special case where all primitive sets are of size 1 (or equivalently, this is the special case where no ancillary qubit will participate in more than a single two-qubit gate), which corresponds directly to fault-tolerant Shor-style syndrome extraction~\cite{shor1997faulttolerant}. Whereas before we were trying to optimize circuits that use one ancillary qubit per stabilizer check, we now model the circuit that uses $w$ ancillary qubits for a weight $w$ stabilizer check (see Fig.~\ref{fig:syndrom_circuits}b). There are some state-preparation concerns for the ancilla involved in this technique, as discussed in Appendix~\ref{sec:cat_prep}, but they are of no consequence to what we discuss here.

This leaves us with just the optimization of two qubit gates between the data and ancillary qubits as in the previous sections. This cost of extra qubit overhead is balanced not only by guarantees of fault tolerance, but also by additional structure that will allow us to compile the code to run in some cases far fewer shuttles than the naive case. We show in Section~\ref{sec:ldpc} that column-regular codes always compile optimally using this syndrome extraction technique.

\subsection{Special Case of Qubit Re-indexing}
As we were doing in the naive syndrome-extraction circuit case, we again only re-index ancillary qubits, so that the data qubits retain the same labeling between the $X$ and $Z$ circuits. However unlike before, we conjecture that in the case of Shor syndrome circuits, re-indexing the data qubits will never lead to an advantage in shuttle minimization.

We now note how the structure of Shor-syndrome circuits causes our problem definition to change. Because all primitive sets become size 1 (i.e.\ an ancillary qubit interacts with exactly one data qubit, ensuring fault-tolerance), instead of finding an assignment from integers to sets, we now only need to find an assignment between integers and integers. More formally, making use of all sets being singletons, we can simplify the problem formulation from the naive syndrome case to:\\

\textbf{Shor Circuit Problem Formulation}: Given a multi-set of $s$ integers $\mathcal{S} = \{p_1,p_2,...,p_s\}$, find a bijection $\pi: [s] \rightarrow [s]$, such that $|\mathcal{G}|_0 $ is minimized, which can be written as follows: 
\begin{equation}
    \min_\pi |\mathcal{G}|_0  = \min_\pi \Bigl| \bigcup_{i=1}^s \{p_i+ \pi(i) \}\Bigr|
\end{equation}
Just as in the naive syndrome case, the minimization of $|\mathcal{G}|_0 $ corresponds to the minimization of the number of required number of shuttles, and the $\pi$ we find in this process describes how we should reorder our ancillary qubits. As before, doing so also increases efficiency and decreases the amount of time required per syndrome extracting cycle. Refer to Fig.~\ref{fig:shor_reindex} for an example.\\

 \begin{figure}
    \centering
    \includegraphics[width=\linewidth]{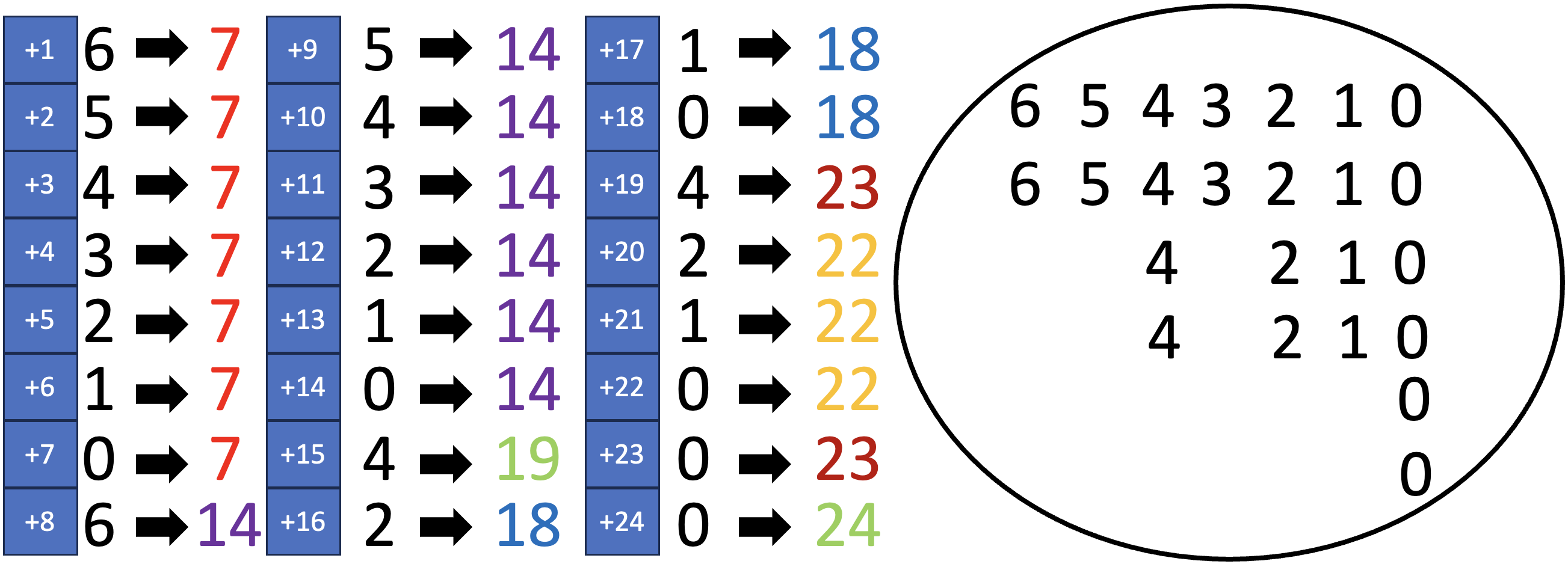}
    \caption{\textbf{Feasible solution to re-indexing of Shor-style syndrome extraction for the Steane code}. On the right is our set $\mathcal{S}$, derived from the Shor-style syndrome extraction circuit for the Steane code (X and Z checks are compiled together to create a more pedgocical example). On the left we have a non-optimal ordering. Here we are placing the numbers from the circle (each number corresponding to a gate in the syndrome circuit) onto this "ladder", with the goal being to minimize the number of different output values obtaining by the addition of the placed value with the index it was placed in. Each output value corresponds to the $\delta$ of the placed gate after re-indexing. The ordering in this figure results in 7 different output (or $\delta$) values, although an optimal solution of 6 values exists, namely Fig.~\ref{fig:steane_chains}.
}
    \label{fig:shor_reindex}
\end{figure}

 \begin{figure}
    \centering
    \includegraphics[width=\linewidth]{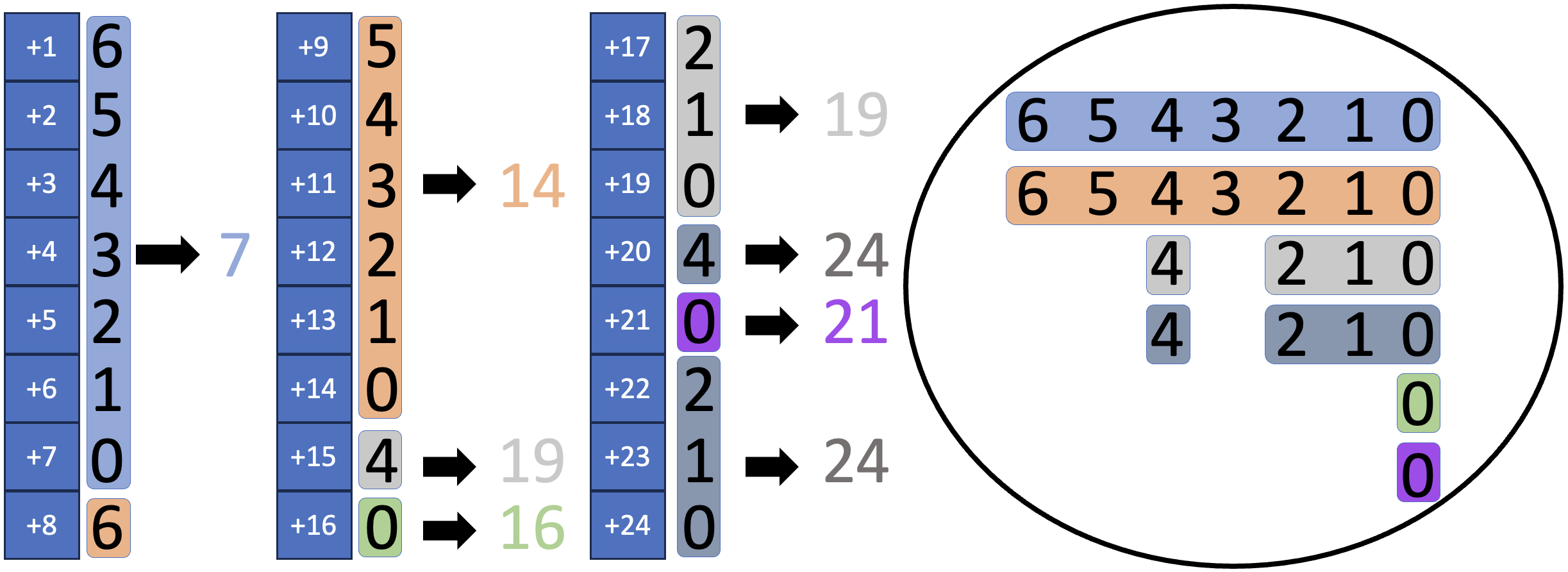}
    \caption{\textbf{Optimal solution to Shor re-indexing on Steane code}. Using  Algorithms 1 and 2, we can find an optimal reordering of qubits for the Steane code's Shor-style syndrome extraction circuit. On the right, we group the elements into six chains, the maximum column weight of the code. On the left, we show how we can pack these without breaking these chains up, giving us an optimal result of six shuttles. Each output value from mapping values to the ladder corresponds to a particular configuration of the $2 \times n$ architecture. Note: as with Fig.~\ref{fig:shor_reindex}, the $X$ and $Z$ checks are presented together for pedagogical reasons.
}
    \label{fig:steane_chains}
\end{figure}

This new formulation allows us to make a few claims about the hypothetical optimal solution structure, which we can leverage in the design of our algorithm.
First, we notice that if a value $v$ occurs $f$ times in $\mathcal{S}$, then $\min_\pi |\mathcal{G}|_0 \geq f$. This follows simply from us considering a sub-multiset $U := \{v, v, ... v \} \subseteq \mathcal{S}$, and $U$ contains $f$ elements. Now, if we add to each element of $U$ a value from $\{1,2,...,s\}$ without repetition, and call the resulting multiset $U'$, then $U'$ will contain exactly $f$ unique elements, regardless of how we choose elements from $\{1,...,s\}$. 

Therefore, we can extrapolate this result to saying that a lower bound on the minimum number of shuttles will be the maximum number of times any number appears in $\mathcal{S}$, which corresponds to the maximum column weight of a code's stabilizer tableau (Section~\ref{sec:ldpc} elaborates on this point). Referring to Fig.~\ref{fig:shor_reindex}, one might ascertain that an optimal solution would be some collection of sequences that together map to a single value, that is, mapping $\{6,5,4,3,2,1,0\}$ to $\{u,u+1,u+2,u+3,u+4,u+5,u+6\}, u\in \mathbb{Z}^+$, would result in a set with only one unique element (and therefore the corresponding gates can then run concurrently one the same $2 \times n$ architecture configuration). Moreover, these sequences need not be strictly consecutive; there may be gaps in them, such as mapping $\{4,2,1,0\}$ to $\{v, v+2, v+3, v+4\}, v\in \mathbb{Z}^+$. We call these sets derived from $\mathcal{S}$ "chains", and discontinuities in their elements "gaps". In other words, a \emph{chain} is a collection of distinct elements that can all be easily mapped to the same value when placed on the ladder (the "ladder" is functionally the same as the "staircase" in the previous section, see Fig.~\ref{fig:ladder_staicase_translation} for further clarification on how to translate between these two notations). When we interpret this back to being a circuit, it will be a collection of gates that will run at the same time step, as they will all have the same $\delta$ after re-indexing.

Therefore, the first step of our algorithm is to find a minimal set of chains, as shown in Algorithm~\ref{alg:making_chains} (\emph{a} minimal set because there may be multiple minimal sets of chains possible for a given input). We do this by initializing an empty set that will become our first chain and then begin moving values from $\mathcal{S}$ to this set, never adding a value that is already present in the chain. When we can no longer add any more values, we create a new empty set to be the second chain and repeat this process. Then we create a third chain and so on until there are no elements left in $\mathcal{S}$.

\begin{algorithm}[H]
\caption{Generation of a Minimal Set of Chains}\label{alg:making_chains}
\begin{algorithmic}[H]
    \Require An array of integers $X$, corresponding to the primitive sets.
    \Ensure A 2D array $Y$, corresponding to the minimal set of chains 
    \State Sort $X$ from largest to smallest\\
    \State Calculate $\text{max\_value}$ to be the largest value in $X$\\
    \State Calculate $\text{num\_chains}$ to be the highest number of times any given element appears in $X$.\\
    
    \State Initialize $Y$ to a 2D array of size $\text{num\_chains} \times \text{max\_value}$, each element is equal to $-1$.
    \For{$x$ in $X$}
    \State $\text{index} \gets$ index($x$) \Comment{get index in $Y$ that corresponds to the value $x$}
    \State $\text{placed} \gets $ False  
    \For{i in $1:\text{num\_chains}$}
    \If{$Y[i][\text{index}] = -1$ AND NOT placed} 
        \State $Y[i][\text{index}] \gets x$
        \State placed $\gets$ True
    \EndIf
    \EndFor
    \EndFor
    \State Depending on implementation, one may have some chains with leading or trailing $-1$'s. If desired, remove those.\\
    \Return $Y$
\end{algorithmic}
\end{algorithm}

Note that this construction will always result in a minimal set of chains with a size equal to the maximum column weight. This is because the most common element in $\mathcal{S}$ will necessitate a new chain for each copy of itself. A minimal set of chains for the Steane code is presented in the right-hand portion of Fig.~\ref{fig:steane_chains}.

Once we have this set of chains, we want to place them in a space of size $|\mathcal{S}|$, making as few cuts to the set of chains as possible, as each chain corresponds to a set of gates that can run in parallel. As we show in Appendix~\ref{sec:NP_proof}, this is NP hard. Therefore, we consider a greedy method where we always place the largest chain (not counting the gaps in the chain) in the first possible position. If a chain does not fit, then we break it into two chains, and return those chains to our set of chains, and continue to place the largest chain remaining in our set. We heuristically break chains with gaps at the gap that is closest to either of the two ends of the chain. If a chain has no gaps, then we break off a single link. The pseudo-code for this is outlined in Algorithm~\ref{alg:sssc}, for an illustration of how this transforms a circuit refer to Fig.~\ref{fig:compilation_comparison}, and an implementation in Julia is available on Github \footnote{\url{https://github.com/amicciche/CircuitCompilation2xn}}~\cite{repo2xn}. The left-hand portion of Fig.~\ref{fig:steane_chains} shows this method resulting in an optimal assignment for the Steane code.

\begin{algorithm}[H]
\caption{Shor Syndrome Specialized Compilation}\label{alg:sssc}
\begin{algorithmic}[H]

\Require A minimal set of chains $Y$, which can be generated via Algorithm~\ref{alg:making_chains}. This is an array which consists of num\_chains rows, which each contains it's own array corresponding to the links and gaps in the chain. -1 denotes a gap in the chain, while other numbers correspond to primitive set values.
\Ensure A sorted list of numbers whose values are the primitive sets, and whose ordering is the order that when $[1,...s]$ is added, results in an array with as ideally as few  distinct values as possible.\\

\State sort($Y$) $\gets$ sorts Y by the total of number of non -1 elements in each chain, from largest to smallest.\\
\State ladder $\gets$ array of size equal to the total number of non -1 elements over all chains in $Y$, which is the size of our original input, $X$, in Algorithm~\ref{alg:making_chains}. Initialize all elements to -1.\\
\State pool $\gets Y$ 
\State sort(pool)
\While{pool is not empty}
\State chain $\gets$ pop from front of pool. \Comment{This should be the chain with the most non -1 values} \\
\State Try to find the first position in ladder such that the non-negative elements in chain line up with -1 elements in ladder. That is, all elements in chain can be assigned sequentially to ladder, starting at the found position, without overwriting any non-negative values.

\If{such a position was found}
    \State Assign the values in chain to ladder starting from that position.
\Else 
    \State Split chain into two smaller chains, chain\_a and chain\_b. If the distance from the beginning of chain to its first -1 is shorter than the distance from the end of chain to its last -1, then split at the first -1, otherwise split at the last -1. (removing the trailing or leading -1). If there are no -1's, simply let chain\_a be the first element, and chain\_b be the rest. 
    \State push chain\_a and chain\_b into pool
    \State sort(pool)
\EndIf
\EndWhile \\
\Return ladder
\end{algorithmic}
\end{algorithm}

 \begin{figure}
    \centering
    \includegraphics[width=\linewidth]{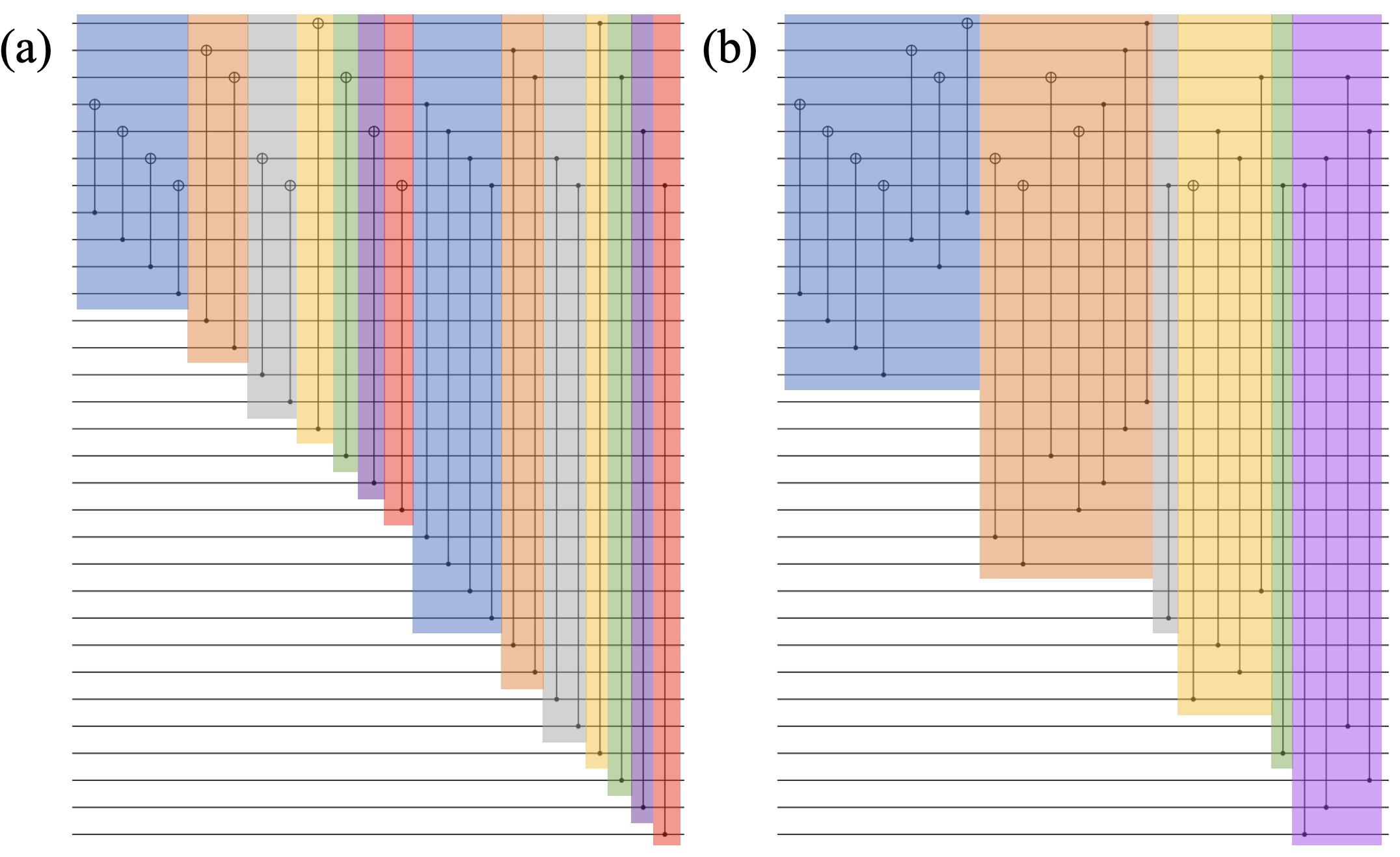}
    \caption{\textbf{Steane code's shor-style syndrome extraction circuit's two qubit gates, both before and after Specialized Shor Syndrome Compilation}. Consecutive gates that are in a color-coded group together can run concurrently on the $2\times n$ architecture (repeated colors have no relation to each other). Furthermore, each grouping corresponds to a particular value of $\delta$. \textbf{a.)} The uncompiled (yet gate shuffled) circuit. There are 14 values for delta in this circuit and therefore 14 shuttles are require to execute this. \textbf{b.)} The circuit after our Specialized Shor Syndrome Compilation algorithm. Notice that now there are now only 6 values of $\delta$, and therefore this circuit will only require 6 shuttles to execute.
}
    \label{fig:compilation_comparison}
\end{figure}

\section{Compilation Results\label{sec:comp_results}}
Tables~\ref{tab:shor_XZ} and ~\ref{tab:naive_XZ} host the number of shuttles required for extracting both the X and Z syndromes of various CSS codes on Shor and naive syndrome circuits respectively. Furthermore, in Appendix.~\ref{sec:sim_results}, we show through gate-by-gate simulation how the logical error rate decreases with compilation on the toric code.

\subsection{Codes Tested}
All codes tested were generated with QuantumClifford.jl \cite{quantumclifford}. Specifically, following the QuantumClifford's ECC API, stablizer tableau were obtained, and then the appropriate methods were called to isolate the X and Z checks, and then to generate the syndrome extracting circuit to be compiled. 
\begin{itemize}
    \item 2-dimensional Toric and Surface codes: $[[18, 2, 3]]$, $[[18, 2, 3]]$,  $[[50, 2, 4]]$ 2D toric codes \cite{Kitaev1997} and $[[13, 1, 3]]$, $[[25, 1, 4]]$, $[[41, 1, 5]]$ 2D surface codes \cite{bravyi1998quantumcodeslatticeboundary} with lattice sizes ranging from $3 \times 3$ to $5 \times 5$ were tested. 
    \item Shor's [[9, 1, 3]] code \cite{shor9}
    \item Steane [[7, 1, 3]] code \cite{steane7}
    \item 2D Color codes: The triangular planar codes with 4.8.8 tiling, given by the instances $[[17,1,5]]$, $[[31,1,7]]$, and $[[71,1,11]]$, and those with 6.6.6 tiling, given by $[[19,1,5]]$, $[[37,1,7]]$, and $[[91,1,11]]$, are considered \cite{landahl2011color}.
    \item Lifted Product (LP) codes: A $[[882,24,d \leq 24]]$ code from Appendix B of \cite{Roffe2023biastailoredquantum}, and a $[[175,19, \leq 10]]$ code Appendix A of \cite{Raveendran2022finiterateqldpcgkp}.
    \item Generalized Bicycle codes: The $[[48,6,8]]$ and $[[254,28, 14 \leq d \leq 20]]$ code instances from Appendix B of \cite{Panteleev2021degeneratequantum}.
    \item Two Block Group Alegbra (2BGA) codes: The $[[40, 8, 5]]$, $[[56, 28, 2]]$, and $[[70, 8, 10]]$ abelian 2BGA codes from Table I of \cite{LinPryadko}, as well as $[[64, 8, 8]]$ non-abelian 2BGA code from Table II of the same work \cite{LinPryadko}.
    \item Multivariate Bicycle code: A weight-$6$ $[[48,4,6]]$ code from \cite{voss2025}.
    \item Bivariate Bicycle codes: The $[[144,12,12]]$, $[[360,12,\leq 24]]$, and $[[756,16,34]]$, as reported in Table~3 of \cite{IBM_gross}.
    \item Trivariate Tricycle codes: The $[[72,6,6]]$, $[[432,12,12]]$, $[[840,9,16]]$ and $[[1029,18,16]]$, as reported in Table~V of \cite{jacob2025single}.
    \item La-Cross codes: The instances $[[100,4,5]]$, $[[117,9,4]]$, $[[296,16,7]]$, and $[[400,16,8]]$ from the $[[(n-k)^2+n^2,k^2,d]]$ La–Cross code family with \textbf{open} boundary conditions, together with the instance $[[98,18,4]]$ from the $[[2n^2,2k^2,d]]$ La–Cross code family with \textbf{periodic} boundary conditions, are reported in Fig.~3 of~\cite{pecorari2025high}.
    \item Quasi-cyclic Generalized Hypergraph Product code: We consider the code instance $[[882,48,16]]$, in which half of the columns have weight~3 and the remaining half have weight~5, as reported in~\cite{Panteleev2021degeneratequantum}. Building on this construction, we generate a new instance with parameters $[[882,12,d]]$, where half of the columns have weight~3 and the other half have weight~7. This construction employs a sparse quasi–cyclic polynomial matrix $A$ over the quotient ring $R=\mathbb{F}_2[x]/(x^\ell-1)$, with each entry being either $0$ or a monomial $x^{i}$ for $0\le i<\ell$, where $\ell$ denotes the circulant size. Each polynomial entry corresponds to a binary $\ell\times\ell$ circulant block. To ensure sparsity, each row and column of $A$ contains at most $w$ nonzero entries, so that the resulting parity–check matrices are $w+\deg b(x)$–limited. A fixed polynomial $b(x)\in R$, referred to as the circulant polynomial, is used to construct the off-diagonal blocks. For the present code instance, we take $\ell=63$ and construct a matrix $A\in M_{n\times n}(R)$ with $n=7$, i.e., $A$ is a $7\times 7$ matrix over $R$:
    \begin{equation*}
\small 
A=\begin{pmatrix}
x^{56} & x^{12} & x^{52} & x^{13} & x^{39} & x^{11} & x^{61} \\
x^{61} & x^{56} & x^{12} & x^{52} & x^{13} & x^{39} & x^{11} \\
x^{11} & x^{61} & x^{56} & x^{12} & x^{52} & x^{13} & x^{39} \\
x^{39} & x^{11} & x^{61} & x^{56} & x^{12} & x^{52} & x^{13} \\
x^{13} & x^{39} & x^{11} & x^{61} & x^{56} & x^{12} & x^{52} \\
x^{52} & x^{13} & x^{39} & x^{11} & x^{61} & x^{56} & x^{12} \\
x^{12} & x^{52} & x^{13} & x^{39} & x^{11} & x^{61} & x^{56}
\end{pmatrix}
b(x)=1+x+x^{6}\in R
\end{equation*}
The code is then defined via the block parity–check matrices
\[
H_X=\begin{bmatrix}A, & b(x)I_m\end{bmatrix},\qquad
H_Z=\begin{bmatrix}b(x)^T I_n, & A^{*}\end{bmatrix}
\]
where $A^{*}$ denotes the transpose of $A$ with each entry polynomial-reversed modulo $x^\ell-1$, and $I_m, I_n$ are identity matrices over $R$ \cite{Panteleev2021degeneratequantum}
\item 3D and 4D Surface codes: The instances of the 3D surface code are $[[12,1,2]]$ for $L=2$ and $[[51,1,3]]$ for $L=3$, while the 4D surface code instances are $[[33,1,4]]$ for $L=2$ and $[[241,1,9]]$ for $L=3$. These codes are constructed via the hypergraph product of $(L-1)\times L$ repetition code chain complexes. The $[[33,1,4]]$ 4D surface code and the construction method for these codes are described in~\cite{PhysRevA.110.062413}.
\item Generalized Toric codes on twisted tori: We consider the code instances $[[292,18,8]]$, $[[392,8,\leq 28]]$, and $[[396,8,\leq 36]]$, as summarized in Tables~III and~IV of \cite{liang2025generalized}.
\item Homological Product codes: We consider the code instances $[[117,9,4]]$ and $[[225,9,6]]$, as reported in Table~III of \cite{xu2025fast}.
\item Double Homological Product codes: We consider the code instances $[[241,1,9]]$ and $[[486,6,9]]$, as reported in Table~I of \cite{Campbell_2019}.
\end{itemize}

\subsection{Shor Compilation\label{sec:shor_comp_description}}
The number of shuttling operations to carry out the fault-tolerant Shor syndrome extracting X and Z circuits for these codes under various levels of compilation is displayed in Table~\ref{tab:shor_XZ}.  The first number in the parenthesis corresponds to the value for the X check syndrome circuit, and the second value corresponds to the Z circuit. Levels of compilation are defined as follows:

\begin{itemize}
    \item \textbf{Uncompiled}: The number of shuttles to run the original circuit provided by QuantumClifford.jl, reading the circuit from left to right, shuttling whenever is needed to run the next gate.
    \item \textbf{GateShuffled}: The number of shuttles required after $\delta$ sorting the circuit as discussed in Section~\ref{sec:gate_shufle}.
    \item \textbf{AHR}: The number of shuttles after reindexing the ancilla by the simple sorting heuristics discussed in Section~\ref{sec:heuristics}.
    \item \textbf{SSSC}: The number of shuttles after reindexing the ancilla by our specialized Shor syndrome circuit algorithm, discussed in Section~\ref{sec:shor}.
    \item \textbf{NumChains \& Blanks}: The "NumChains" value corresponds to the minimum number of shuttles possible, achievable by inserting "blanks" into the ancilla array. The "Blanks" value is then a naive upper bound on the number of blanks required to achieve this. It was calculated by counting the number of gaps in the chains. Chain were introduced in Section~\ref{sec:shor}, and the idea of adding blanks to achieve the minimum number of possible shuttles is discussed further in Section~\ref{sec:ldpc} 
\end{itemize}

In Table~\ref{tab:shor_XZ}, we can see that our techniques find quite good shuttling schedules, trivially so on codes with regular column weight (and so SSSC does not add anything). For codes with non-regular column weight, with only some exception, SSSC does further reduce the number of shuttles compared to our heuristic sorting methods. Furthermore notice that we need not use the same reindexing technique for both the $X$ and $Z$ circuits. For example on the $[[117,9,4]]$ homological product code tested, we could use AHR to compile the X circuit to run in 12 shuttles, while using SSSC to compile the Z circuit to run in 7 shuttles. This is also exemplified in the $[[117,9,4]]$ La-Cross code, which can run with 5 shuttles on the X circuit after compiling via AHR, and then 5 shuttles on the Z circuit after compiling it with SSSC.

Another thing to notice is that since these methods are taking as input a circuit algorithmically generated via its stabilizer tableau, codes like the 2D toric code compile to different shuttle numbers for its X and Z circuits. This is simply due to the stabilizer tableaux for its X and Z checks being literally different due to how they are generated in QuantumClifford.jl \cite{quantumclifford}. The manifestation of the 2D toric code's stabilizer generators form a column-irregular matrix, at least as defined in QuantumClifford.jl, and so optimal compilation is not promised. An interesting avenue for future work might be on techniques that use Gaussian elimination on the stabilizer tableau to work in tandem with our methods which merely permute it.

\subsection{Naive Compilation}
The number of shuttling operations to carry out the naive syndrome extracting circuits for these codes are displayed in Table~\ref{tab:naive_XZ}. The definitions of the types of compilation are the same as for Shor syndrome extraction. Naturally, there are fewer, as we developed some specialized methods for Shor syndrome extraction. For most of the codes tested, AHR does not reduce the number of shuttles for their naive syndrome circuit. This is due to naive compilation being a much harder problem than Shor syndrome compilation, which we show is NP-hard in Appendix~\ref{sec:NP_proof}. The AHR methods are after all heuristics and need not be used if the original ordering of ancilla is better. Development of algorithms to improve performance of more general or other special classes of syndrome extraction circuits (flag qubit circuits are perhaps more interesting than these naive circuits), we leave to future research.

\begin{table*}[]
\begin{tabular}{|l|l|l|l|l|l|l|l|l|l|}
\hline
Code & N & K & D & Uncompiled   & GateShuffled & AHR        & SSSC       & NumChains & Blanks       \\
\hline
2DTORIC                  &   18 &  2 &  3                & (27, 26)     & (21, 21)     & (4, 5)     & (4, 4)     & (2, 2)    & (2, 3)       \\
2DTORIC                  &   32 &  2 &  4                & (49, 48)     & (35, 35)     & (4, 5)     & (4, 5)     & (2, 2)    & (2, 3)       \\
2DTORIC                  &   50 &  2 &  5                & (77, 76)     & (62, 61)     & (4, 5)     & (4, 5)     & (2, 2)    & (2, 3)       \\
2DSURFACE                &   13 &  1 &  3                & (18, 14)     & (14, 13)     & (3, 5)     & (3, 4)     & (2, 2)    & (3, 5)       \\
2DSURFACE                &   25 &  1 &  4                & (36, 30)     & (26, 25)     & (3, 6)     & (3, 4)     & (2, 2)    & (4, 7)       \\
2DSURFACE                &   41 &  1 &  5                & (60, 52)     & (48, 34)     & (3, 7)     & (3, 5)     & (2, 2)    & (5, 9)       \\
SHOR                     &    9 &  1 &  3                & (2,  4)      & (2, 4)       & (2, 4)     & (2, 4)     & (2, 2)    & (0, 4)       \\
STEANE                   &    7 &  1 &  3                & (7, 7)       & (7, 7)       & (3, 3)     & (3, 3)     & (3, 3)    & (1, 1)       \\
2DCOLOR4.8.8             &   17 &  1 &  5                & (19, 19)     & (15, 15)     & (7, 7)     & (4, 4)     & (3, 3)    & (7, 7)       \\
2DCOLOR4.8.8             &   31 &  1 &  7                & (38, 38)     &  (27, 27)    & (9, 9)     & (7, 7)     & (3, 3)    & (10, 10)     \\
2DCOLOR4.8.8             &   71 &  1 & 11                & (94, 94)     & (67, 67)     & (13, 13)   & (9, 9)     & (3, 3)    & (18, 18)     \\
2DCOLOR6.6.6             &   19 &  1 &  5                & (22, 22)     & (17, 17)     & (7, 7)     & (6, 6)     & (3, 3)    & (5, 5)       \\
2DCOLOR6.6.6             &   37 &  1 &  7                & (48, 48)     & (32, 32)     & (10, 10)   & (7, 7)     & (3, 3)    & (9, 9)       \\
2DCOLOR6.6.6             &   91 &  1 & 11                & (127, 127)   & (82, 82)     & (16, 16)   & (12, 12)   & (3, 3)    & (17, 17)     \\
LIFTEDPRODUCT            &  882 & 24 & $ d \leq$ 24      & (2212, 2212) & (1820, 1820) & (3, 3)     & (3, 3)     & (3, 3)    & (0, 0)       \\
LIFTEDPRODUCT            &  175 & 19 & $ d \leq$ 10      & (582, 584)   & (395, 377)   & (4, 4)     & (4, 4)     & (4, 4)    & (0, 0)       \\
GENERALIZEDBICYCLE       &   48 &  6 &  8                & (192, 192)   & (139, 139)   & (4, 4)     & (4, 4)     & (4, 4)    & (0, 0)       \\
GENERALIZEDBICYCLE       &  254 & 28 & $14 \le d \le 20$ & (1144, 1144) & (827, 827)   & (5, 5)     & (5, 5)     & (5, 5)    & (0, 0)       \\
TWOBLOCKGROUPALGEBRA     &   70 &  8 & 10                & (212, 212)   & (145, 145)   & (4, 4)     & (4, 4)     & (4, 4)    & (0, 0)       \\
TWOBLOCKGROUPALGEBRA     &   40 &  8 &  5                & (121, 121)   & (74, 74)     & (6, 6)     & (6, 6)     & (6, 6)    & (0, 0)       \\
TWOBLOCKGROUPALGEBRA     &   56 & 28 &  2                & (142, 142)   & (136, 136)   & (6, 6)     & (6, 6)     & (6, 6)    & (0, 0)       \\
TWOBLOCKGROUPALGEBRA     &   64 & 8 &  8                 & (224, 224)   & (159, 168)   & (6, 6)     & (6, 6)     & (6, 6)    & (0, 0)        \\
MULTIVARIATEBICYCLE      &   48 &  4 &  6                & (144, 144)   & (93, 93)     & (4, 4)     & (4, 4)     & (4, 4)    & (0, 0)       \\
LA-CROSS                 &  100 &  4 &  5                & (192, 156)   & (135, 99)    & (5, 19)    & (12, 15)   & (3, 3)    & (24, 45)     \\
LA-CROSS                 &   98 & 18 &  4                & (252, 252)   & (201, 201)   & (3, 3)     & (3, 3)     & (3, 3)    & (0, 0)       \\
LA-CROSS                 &  117 &  9 &  4                & (240, 216)   & (180, 160)   & (5, 21)    & (12, 5)    & (3, 3)    & (45, 77)     \\
LA-CROSS                 &  296 & 16 &  7                & (630, 580)   & (461, 458)   & (5, 31)    & (46, 27)   & (3, 3)    & (98, 163)    \\
LA-CROSS                 &  400 & 16 &  8                & (876, 816)   & (618, 605)   & (5, 35)    & (52, 39)   & (3, 3)    & (112, 187)   \\
BIVARIATEBICYCLE         &  144 & 12 & 12                & (372, 372)   & (229, 229)   & (3, 3)     & (3, 3)     & (3, 3)    & (0, 0)       \\
BIVARIATEBICYCLE         &  360 & 12 & $d \leq$24        & (930, 930)   & (549, 549)   & (3, 3)     & (3, 3)     & (3, 3)    & (0, 0)       \\
BIVARIATEBICYCLE         &  756 & 16 & $d \leq$34        & (2268, 2268) & (1442, 1442) & (3, 3)     & (3, 3)     & (3, 3)    & (0, 0)       \\
TRIVARIATETRICYCLE       &   72 &  6 &  6                & (207, 415)   & (148, 315)   & (3, 6)     & (3, 6)     &  (3, 6)   & (0, 0)      \\
TRIVARIATETRICYCLE       &  432 & 12 & 12                & (1266, 2532) & (838, 1722)  & (3, 6)     & (3, 6)     &  (3, 6)   &  (0, 0)     \\
TRIVARIATETRICYCLE       &  840 &  9 & 16                & (2520, 5040) & (1642, 3214) & (3, 6)     & (3, 6)     &  (3, 6)   &  (0, 0)     \\
TRIVARIATETRICYCLE       & 1029 & 18 & 16                & (3087, 6174) & (1890, 4429) & (3, 6)     & (3, 6)     & (3, 6)    &  (0, 0)     \\
3DSURFACE                &   12 &  1 &  2                & (9, 23)      & (9, 16)      & (3, 4)     & (3, 4)     & (2, 3)    &  (4, 4)     \\
3DSURFACE                &   51 &  1 &  3                & (66, 138)    & (46, 90)     & (13, 10)   & (10, 8)    &  (2, 4)   & (17, 31)    \\
4DSURFACE                &   33 &  1 &  4                & (64, 75)     & (49, 54)     & (8, 8)     & (7, 7)     & (4, 4)    & (16, 16)    \\
4DSURFACE                &  241 &  1 &  9                & (656, 715)   & (419, 461)   & (67, 37)   & (44, 27)   & (4, 4)    &  (152, 196) \\
GENERALIZEDTORIC         &  292 & 18 &  8                & (730, 730)   & (504, 504)   & (3, 3)     & (3, 3)     & (3, 3)    & (0, 0)        \\
GENERALIZEDTORIC         &  392 &  6 & $d \leq$28        & (280, 280)   & (196, 196)   & (3, 3)     & (3, 3)     & (3, 3)    & (0, 0)        \\
GENERALIZEDTORIC         &  396 &  8 & $d \leq$36        & (991, 991)   & (583, 583)   & (3, 3)     & (3, 3)     & (3, 3)    & (0, 0)       \\
HOMOLOGICALPRODUCT       &  117 &  9 &  4                & (225, 225)   & (155, 162)   & (12, 8)    & (22,7)     & (3, 3)    & (27, 42)       \\
HOMOLOGICALPRODUCT       &  225 &  9 &  6                & (618, 630)   & (432, 447)   & (27, 14)   & (28, 11)   & (4, 4)    & (72, 120)         \\
DBLEHOMOLOGICALPRODUCT   &  486 &  6 &  9                & (1782, 1782) & (1204, 1204) &  (4, 4)    & (4, 4)     &  (4, 4)   &  (0, 0)            \\
QUASICYCLIC\_GHP         &  882 & 48 & 16                & (3093, 3093) & (2151, 2151) &  (5, 5)    &  (5, 5)    &  (5, 5)   &  (0, 0)     \\
QUASICYCLIC\_GHP         &  882 & 12 &                   & (3976, 3976) & (2829, 2829) &  (7, 7)    &  (7, 7)    &   (7, 7)  &  (0, 0)      \\
\hline
\end{tabular}
\caption{\textbf{Compilation numbers for Shor-style syndrome circuits for extracting X, Z stabilizers separately on CSS codes.}  
As described in Section.~\ref{sec:shor_comp_description}, the different columns correspond to the required numbers of shuttles corresponding to different levels of compilation (the number of shuttles is listed for $X$ circuits first and then $Z$). However "NumChains" corresponds to the smallest number of shuttles possible, obtainable by introducing "blanks" into the ancilla array. A very naive upper bound on the number of blanks required for perfect compilation is presented in the "Blanks" column, rather than a number of shuttles. Refer to Section~\ref{sec:ldpc} for more on blanks and perfect compilation. Quantum error correcting code parameters: \textbf{N} is the number of physical qubits, \textbf{K} the number of logical qubits, and true code distance $\textbf{D} = \min(d_X, d_Z)$ is defined as the minimum weight of a non-trivial logical Pauli operator. Distances were calculated using the mixed-integer programming method of \cite{landahl2011color}.}
 \label{tab:shor_XZ}
\end{table*}

\begin{table*}[]
\begin{tabular}{|l|l|l|l|l|l|l|}
\hline
Code       & N & K & D                                     & Uncompiled   & GateShuffled & AHR  \\
\hline
2DTORIC                 &   18 &  2 &  3                    & (32, 32)     & (6, 6)       & (11, 10)   \\
2DTORIC                 &   32 &  2 &  4                    & (60, 60)     & (6, 6)       & (12, 11)   \\
2DTORIC                 &   50 &  2 &  5                    & (96, 96)     & (6, 6)       & (12, 11)   \\
2DSURFACE               &   13 &  1 &  3                    & (20, 20)     & (5, 6)       & (5, 6)     \\
2DSURFACE               &   25 &  1 &  4                    & (42, 42)     & (6, 7)       & (6, 7)     \\
2DSURFACE               &   41 &  1 &  5                    & (72, 72)     & (7, 8)       & (7, 8)     \\
SHOR                    &    9 &  1 &  3                    & (12, 10)     & (8, 4)       & (8, 4)     \\
STEANE                  &    7 &  1 &  3                    & (12, 12)     & (8, 8)       & (7, 7)     \\
2DCOLOR4.8.8            &   17 &  1 &  5                    & (36, 36)     & (13, 13)     & (13, 13)   \\
2DCOLOR4.8.8            &   31 &  1 &  7                    & (72, 72)     & (21, 21)     & (20, 20)   \\
2DCOLOR4.8.8            &   71 &  1 & 11                    & (180, 180)   & (43, 43)     & (38, 38)   \\
2DCOLOR6.6.6            &   19 &  1 &  5                    & (42, 42)     & (13, 13)     & (11, 11)   \\
2DCOLOR6.6.6            &   37 &  1 &  7                    & (90, 90)     & (23, 23)     & (20, 20)   \\
2DCOLOR6.6.6            &   91 &  1 & 11                    & (240, 240)   & (52, 52)     & (47, 47)   \\
LIFTEDPRODUCT           &  882 & 24 & d $\leq$ 24           & (2646, 2646) & (13, 13)     & (81, 32)   \\
LIFTEDPRODUCT           &  175 & 19 & d $\leq$ 10           & (588, 588)   & (46, 53)     & (103, 55)  \\
GENERALIZEDBICYCLE      &   48 &  6 & 8                     & (192, 192)   & (13, 13)     & (38, 36)   \\
GENERALIZEDBICYCLE      &  254 & 28 & 14 $\leq$ d $\leq$ 20 & (1270, 1270) & (18, 18)     & (67, 67)   \\
TWOBLOCKGROUPALGEBRA    &   70 &  8 & 10                    & (280, 280)   & (14, 14)     & (42, 40)   \\
TWOBLOCKGROUPALGEBRA    &   40 &  8 &  5                    & (160, 160)   & (17, 17)     & (41, 31)   \\
TWOBLOCKGROUPALGEBRA    &   56 & 28 &  2                    & (224, 224)   & (15, 15)     & (24, 15)   \\
TWOBLOCKGROUPALGEBRA    &   64 & 8  &  8                    & (256, 256)   & (55, 55)     & (67, 66)   \\
MULTIVARIATEBICYCLE     &   48 &  4 &  6                    & (144, 144)   & (16, 16)     & (35, 23)   \\
LA-CROSS                &  100 &  4 &  5                    & (252, 252)   & (16, 20)     & (16, 20)   \\
LA-CROSS                &   98 & 18 &  4                    & (294, 294)   & (10, 10)     & (37, 28)   \\
LA-CROSS                &  117 &  9 &  4                    & (270, 270)   & (16, 22)     & (16, 22)   \\
LA-CROSS                &  296 & 16 &  7                    & (720, 720)   & (24, 32)     & (24, 32)   \\
LA-CROSS                &  400 & 16 &  8                    & (1008, 1008) & (28, 36)     & (28, 36)    \\
BIVARIATEBICYCLE        &  144 & 12 & 12                    & (432, 432)   & (12, 12)     & (21, 68)            \\
BIVARIATEBICYCLE        &  360 & 12 & d $\leq$ 24           & (1080, 1080) & (12, 12)     & (43, 186)           \\
BIVARIATEBICYCLE        &  756 & 16 & d $\leq$ 34           & (2268, 2268) & (19, 19)     & (39, 58)        \\
TRIVARIATETRICYCLE      &   72 &  6 &  6                    & (216, 432)   & (28, 53)     & (55, 73)          \\
TRIVARIATETRICYCLE      &  432 & 12 & 12                    & (1296, 2592) & (46, 90)     & (86, 217)           \\
TRIVARIATETRICYCLE      &  840 &  9 & 16                    & (2520, 5040) & (43, 84)     & (98, 254)                \\
TRIVARIATETRICYCLE      & 1029 & 18 & 16                    & (3087, 6174) & (39,76)      & (107, 390)            \\
3DSURFACE               &   12 &  1 &  2                    & (16, 27)     & (10, 9)      & (9, 9)          \\
3DSURFACE               &   51 &  1 &  3                    & (84, 150)    & (16, 22)     & (16, 33)         \\
4DSURFACE               &   33 &  1 &  4                    & (83, 84)     & (23, 18)     & (23, 18)        \\
4DSURFACE               &  241 &  1 &  9                    & (758, 759)   & (72, 57)     & (79, 83)          \\
GENERALIZEDTORIC        &  292 & 18 &  8                    & (876, 876)   & (230, 230)   & (189, 201)           \\
GENERALIZEDTORIC        &  392 &  8 & $d \leq$28            & (336, 336)   & (95, 95)     & (95, 93)              \\
GENERALIZEDTORIC        &  396 &  8 & $d \leq$36            & (1188, 1188) & (312, 312)   & (251, 243)                   \\
HOMOLOGICALPRODUCT      &  117 &  9 &  4                    & (225, 225)   & (28, 22)     & (43, 52)              \\
HOMOLOGICALPRODUCT      &  225 &  9 &  6                    & (630, 630)   & (38, 32)     & (49, 114)                \\
DBLEHOMOLOGICALPRODUCT  &  486 &  6 &  9                    & (1944, 1944) & (83, 83)     & (259, 91)                   \\
QUASICYCLIC\_GHP        &  882 & 48 & 16                    & (3528, 3528) & (19, 19)     & (102, 49)         \\    
QUASICYCLIC\_GHP        &  882 & 12 &                       & (4410, 4410) & (31, 31)     & (145, 126)                       \\
\hline
\end{tabular}
\caption{\textbf{Different compilations of the naive syndrome circuit for extracting all stabilizers.} Same as Table~\ref{tab:shor_XZ}, except for naive syndrome circuits. Naturally there are fewer columns, as we developed additional methods specifically for Shor-syndrome extraction.}
\label{tab:naive_XZ}
\end{table*}

\section{Optimal Compilation of regular LDPC codes\label{sec:ldpc}}
So far we have discussed techniques for general compilation, the special case of Shor-syndrome extraction, as well as simulation of how these techniques can lower pseudo-thresholds. In this section, we illuminate a family of LDPC codes that can always be reordered in polynomial time such that the number of resulting shuttles in its Shor syndrome extraction circuit is perfectly minimized.  

First we notice how a stabilizer tableau or parity check matrix maps to the input $\mathcal{S}$ of our ancillary qubit re-indexing problem. For example, take the Steane code (treating the X and Z checks together to create a more pedagogical example):
\begin{equation}
    \begin{bmatrix}
        + & & & & X & X & X & X\\
        + & & X & X & & & X & X\\
        + & X & & X & & X & & X\\
        + & & & & Z & Z & Z & Z\\
        + & & Z& Z & & & Z & Z\\
        + & Z & & Z & & Z & & Z
    \end{bmatrix}
    \rightarrow
    \begin{bmatrix}
        & & & 3 & 2 & 1 & 0\\
        & 5 & 4 & & & 1 & 0\\
        6 & & 4 & & 2 & & 0\\
        & & & 3 & 2 & 1 & 0\\
        & 5& 4 & & & 1 & 0\\
        6 & & 4 & & 2 & & 0
    \end{bmatrix}
\end{equation}

If we were to build a syndrome extracting circuit of this code, we would take each row of the Steane code, attributing the first column to the first data qubit (the top of the circuit diagram), the second column to the second qubit on our diagram and so on, then attach gates to ancilla wherever we find a non-zero entry. For example, the first row would cause four gates to be connected to the fourth, fifth, sixth, and seventh data qubits respectively. When we apply our ancillary qubit re-indexing techniques to this, the ancillary qubits connected to the bottom (seventh) data qubit would have a value of 0 in their primitive set, ancilla connected to the sixth data qubit would have a 1 in their primitive set and so on. Thus, in the case of Shor-style syndrome extraction, our set $\mathcal{S}$ will be comprised of the entries in the right matrix in (3) (while in naive syndrome circuits, the primitive sets will correspond to the rows of the right side of that same matrix).

We now notice that if our error correction code has regular column weight $w_c$ (each column has the same number of non-zero entries), then this would mean that all elements that occur in $\mathcal{S}$, occur the same number of times. Because each element occurs $w_c$ times, we can then make $w_c$ identical chains of the form $n-1, n-2, ..., 0$, each of which has no gaps, and will pack perfectly. Therefore, all LDPC codes with regular column weight can be compiled to run on $2 \times n$ architectures in exactly $w_c$ shuttles. It is worth noting that this would then include all "regular" LDPC codes, as they have constant column and row weight~\cite{eczoo_regular_ldpc}. Furthermore, the $[[144,12,12]]$ gross code ~\cite{IBM_gross} also has regular column weight and would be an excellent choice of error correcting code for this architecture, as well as trivariate tricycle codes \cite{jacob2025single}, generalized toric codes on twisted tori \cite{liang2025generalized}, lifted product codes and its' subfamilies including generalized bicycle and two block group algebra codes \cite{Panteleev2021degeneratequantum}. In addition, we observed some irregular codes that also compiled perfectly such as $[[882, 48, 16]]$, and $[[882, 12, ?]]$ quasi-cyclic generalized hypergraph product codes \cite{Panteleev2021degeneratequantum}.

\subsection{Optimal Compilation on Irregular LDPC codes\label{sec:blank_qubits}}

Until now, we have viewed the problem of re-indexing the ancillary qubits under Shor-style syndrome extraction as a problem that implicitly involves maintaining the tightest possible spacing of the qubits. If we relax this assumption and allow for spaces in the array, then we can trivially achieve optimal compilation in regard to the number of required shuttles by deriving the set of minimal chains, and placing those consecutively on an "extended" staircase to accommodate the gaps in some of the chains. The number of additional quantum dots required to do this will, in the most naive placement scheme, be equal to the sum of the sizes of all gaps in the set of minimal chains. Before, we were maintaining optimal spacing, while trying to minimize shuttles, however one could instead aim to fix the number of shuttles to optimal (equal to the maximum column weight of the code), while minimizing the spacing. We prove the former to be NP hard in Appendix\ref{sec:NP_proof}, while we conjecture the latter to also be NP-hard.

\section{Discussion and Conclusion\label{sec:the_end}}
We have demonstrated methods for scheduling the shuttles of qubits in $2 \times n$ one-dimensional architectures, highlighting a large family of LDPC codes that under Shor-syndrome extraction only ever requires a number of shuttles exactly equal to the code's column weight, which notably includes bivariate bicycle quantum LDPC codes, which were shown in recent work to have a threshold on par with the surface code~\cite{IBM_gross}. Other codes that would make great candidates for this architecture would include: trivariate tricycle codes \cite{jacob2025single}, generalized toric codes on twisted tori \cite{liang2025generalized}, lifted product codes and its' subfamilies including generalized bicycle and two block group algebra codes 
\cite{Panteleev2021degeneratequantum}, as already mentioned in Section~\ref{sec:ldpc}.


 The most natural next problem to address would be the state preparation of "interwoven" cat states needed to perform Shor syndrome extraction. "Interwoven" corresponds to after ancilla reindexing, we might require the ancilla array contain distinct multiple cat states that need to be arranged in the ancilla array such that the physical qubits constituting one of the cat states are not all neighbors to one another (see Fig.~\ref{fig:cat_prep}). Throughout this work, we assumed that we had access to fault-tolerantly prepared cat states to load into the ancilla array, however, the details of how to do and optimize this, especially the process of interweaving, we leave to future work. In Appendix~\ref{sec:cat_prep}, we propose a sketch for how one could start with an array of consecutive cat states and interweave them, using an additional empty quantum dot array.
 
For other future work, we plan to investigate the design of more elegant algorithms for compilation, perhaps exploring the trade-off between minimizing the size of the ancilla array and minimizing the number of shuttles. We \editcolor{also hope to} consider other paradigms and pieces of the quantum error correction pipeline possible in the specific hardware implementation; for instance, incorporating flag qubits, architectures with more than two rails, how to optimize performance of logical operations, and discovery of other family codes that can have their syndrome extracted in ways that require only a few shuttling operations.

Nevertheless, what we have presented here is already a substantial achievement towards embarking on what we believe to be a rich area of research, both from the perspectives of theoretical computer science as well as practical quantum error correction. That is, the rather novel algorithmic questions raised by optimizing error correcting codes to run on practical hardware are challenging enough to be interesting, yet seemingly constrained enough to allow for progress to be made. The culmination of which is the practical result of having more efficient quantum error correction schemes specifically tailored to a specific and promising type of quantum hardware. We hope to see this work expanded upon, as well as other works that explore these sorts of optimizations for other qubit platforms as we continue to approach a future of fault-tolerant quantum computation. \\

\paragraph*{Code availability:} The underlying simulator is available as a registered open source package in the Julia ecosystem: CircuitCompilation2xn.jl~\cite{repo2xn}. The goal of that package is to allow users to provide any stabilizer code's tableau, or any syndrome extraction circuit, and have the package compile it for better shuttling performance on the $2 \times n$ architecture. 

\begin{acknowledgements}
We would like to thank our colleagues Filip Rozepedek and Guus Avis for their fruitful discussions regarding gate commutativity. We also thank Nithin Raveendran for his recommendation to see how bivariate bicycle codes compile on this architecture. We would like to acknowledge colleagues Ferdinand Kuemmeth and Fabrizio Berritta as well. This work was supported by funding from the Inge Lehmann Programme of the Independent Research Fund Denmark, the US Army Research Office (ARO) under Award No. W911NF-24-2-0043, and the NSF under grants 1941583, 2346089, and 2402861.
\end{acknowledgements}

\bibliography{bibliography}


\appendix

\section*{Appendix}

In this appendix, we provide examples of the block stacking visualization on valid error correcting codes, an illustration of how to translate between the "ladder" and "staircase" notations, a proof for the NP hardness of the ancillary qubit re-indexing problem, a discussion of how to prepare the necessary interwoven cat states needed for compiled Shor syndrome circuits, and a presentation of plotting logical versus physical error rates for different compilation schemes, following a circuit-level noise model.

\subsection{More Block Visualization Examples\label{sec:more_block_examples}}
Refer to Figures~\ref{fig:steane_staircase} and~\ref{fig:shor_staircase} to see how the simple heuristics from Section~\ref{sec:heuristics} can be applied to the naive syndrome extracting circuits of Steane's [[7,1,3]] code~\cite{steane7}  and Shor's [[9,1,3]] code~\cite{shor9}. In Section~\ref{sec:shor}, we introduced a condensed version of the staircase notation that is specifically for Shor syndrome circuits. Fig.~\ref{fig:ladder_staicase_translation} illustrates how these two notations are equivalent.

\begin{figure}[h]
    \centering
    \includegraphics[width=\linewidth]{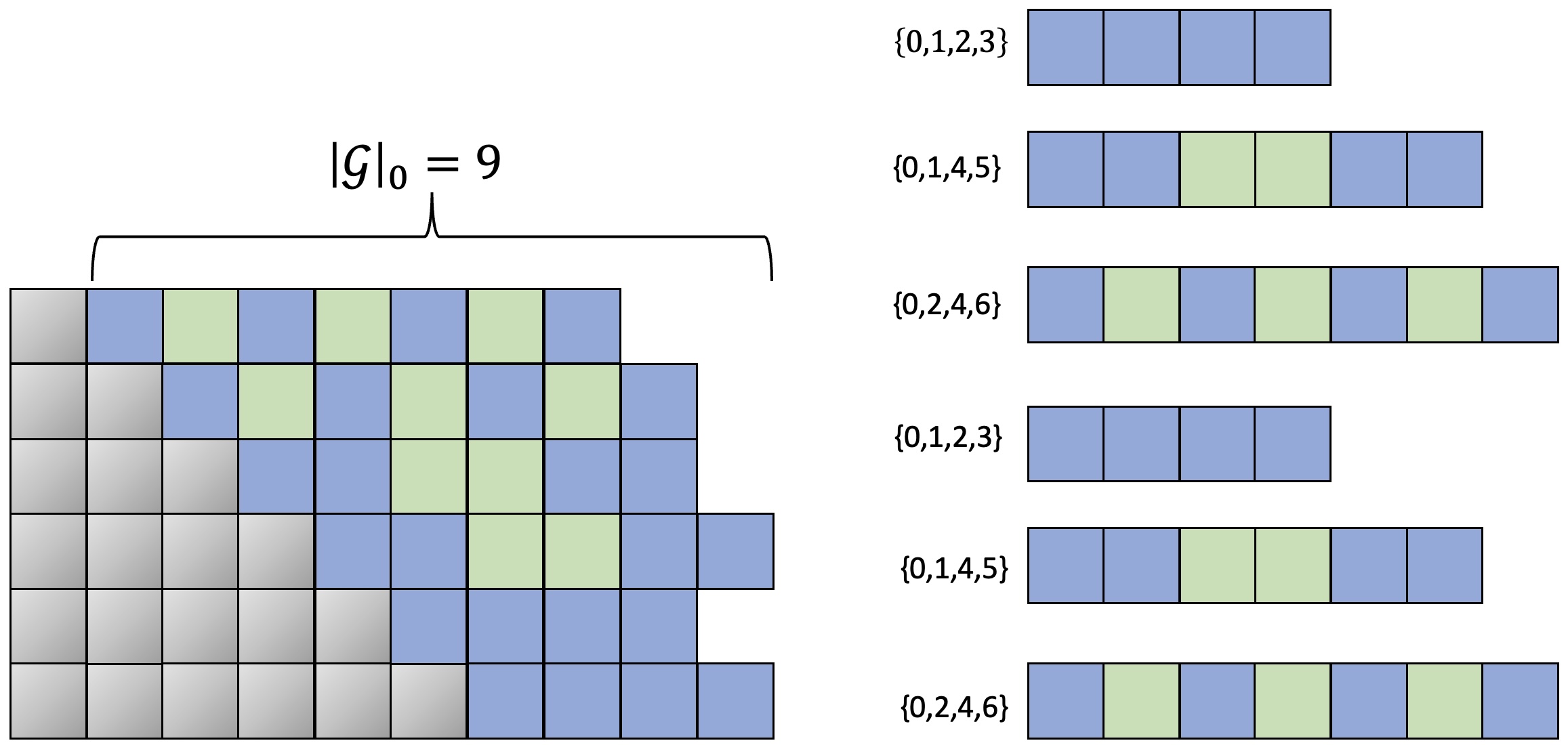}
    \caption{\textbf{Application of heuristics on the Steane code}. Here we have the primitive sets of the naive syndrome extraction circuit of the Steane code expressed in block form. We can see that either of the first two heuristics discussed in Section~\ref{sec:heuristics} will allow use to find an optimal solution for this case of ancillary qubit re-indexing. Note: As this example is pedagogical in nature, this corresponds to reindexing the X and Z syndrome circuits together, something that in practice will be avoided.
}
    \label{fig:steane_staircase}
\end{figure}

\begin{figure}
    \centering
    \includegraphics[width=\linewidth]{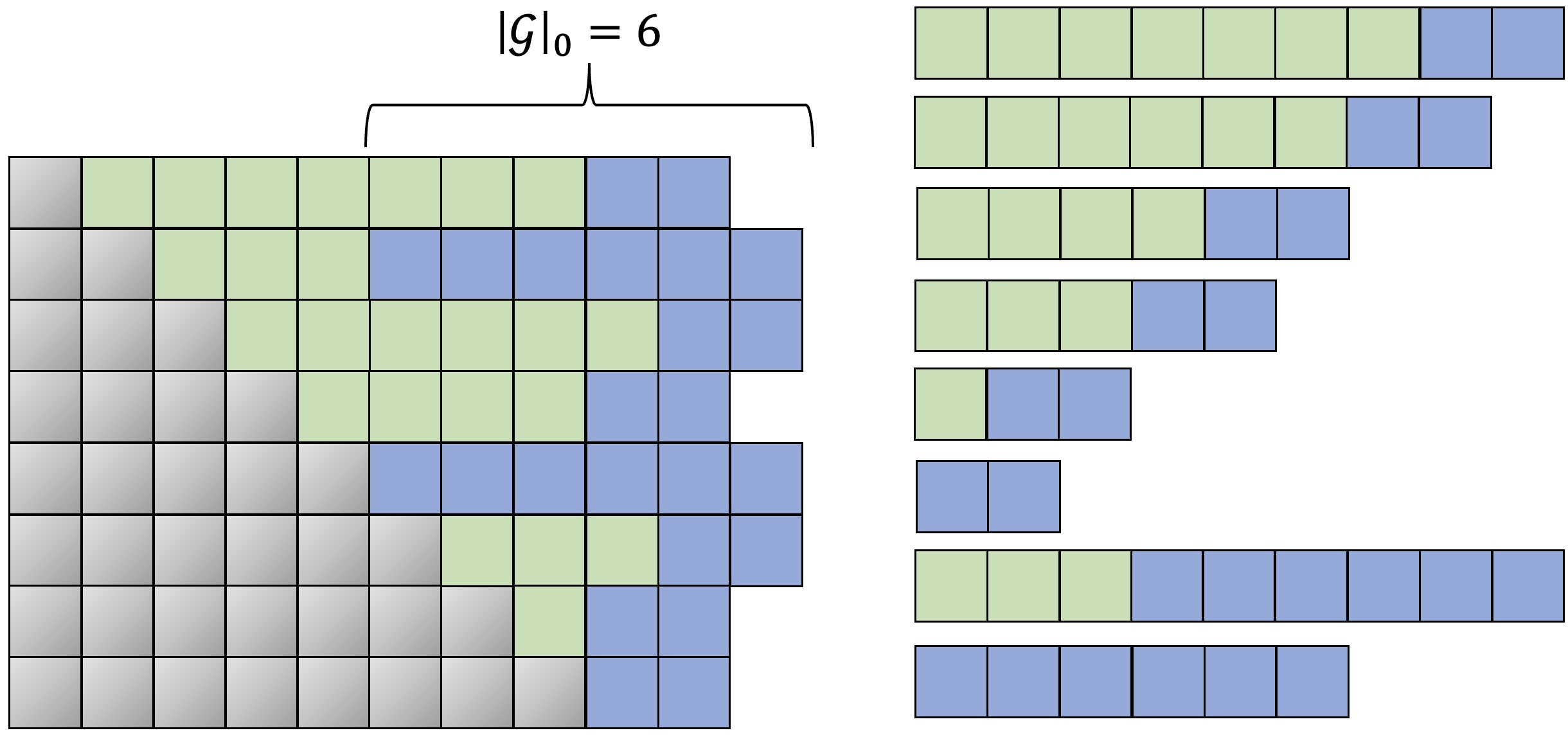}
    \caption{\textbf{Application of heuristics on the Shor code}. Here we have the primitive sets of the naive syndrome extraction circuit of the Shor code expressed in block form. We can see that either of the first two heuristics discussed in Section~\ref{sec:heuristics} will allow use to find an optimal solution for this case of ancillary qubit re-indexing. Note: As this example is pedagogical in nature, this corresponds to reindexing the X and Z syndrome circuits together, something that in practice will be avoided.
}
    \label{fig:shor_staircase}
\end{figure} 

\begin{figure}
    \centering
    \includegraphics[width=\linewidth]{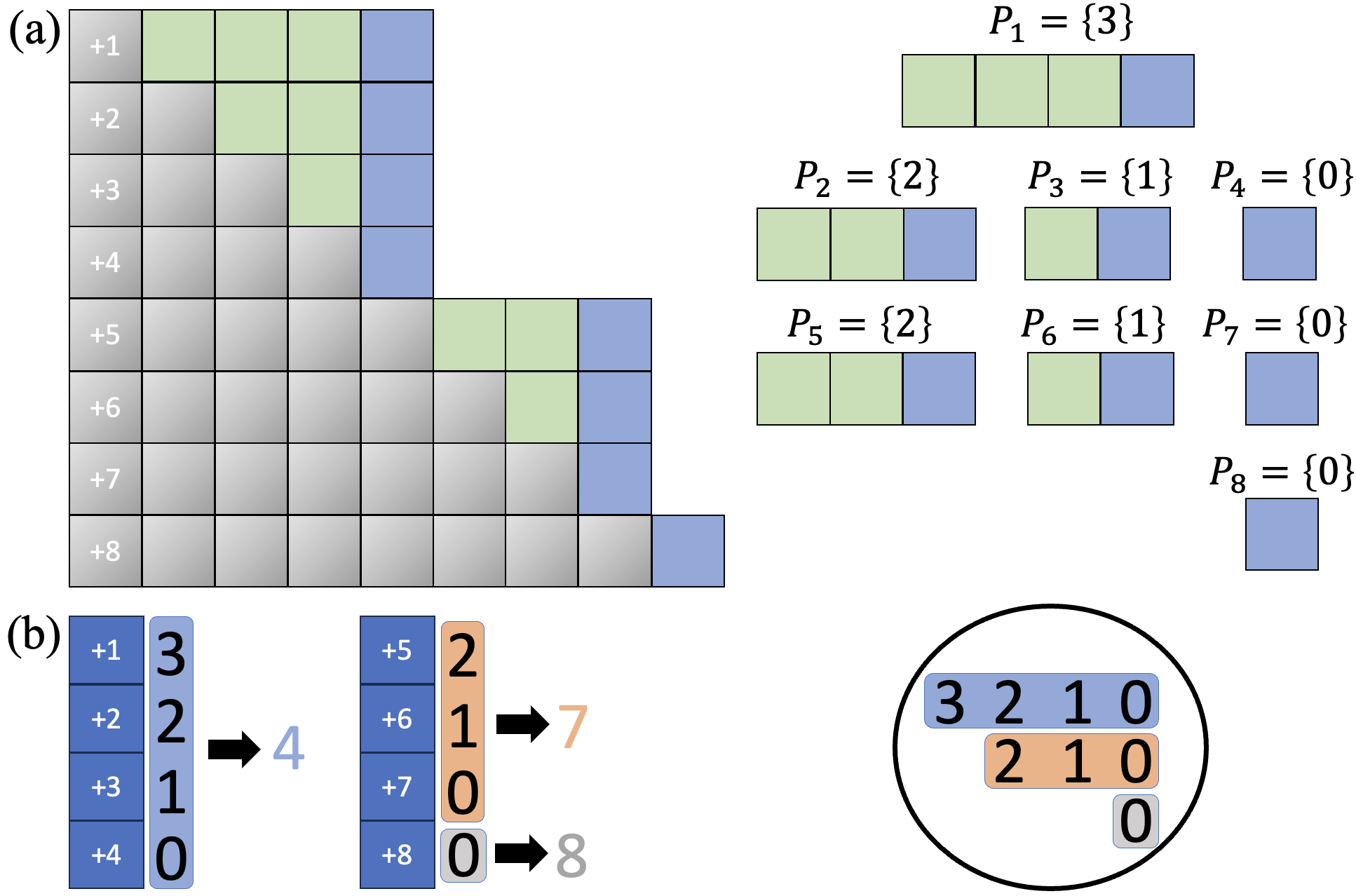}
    \caption{\textbf{Translation between "staircase" and "ladder" notation}. \textbf{a)} The "staircase" packing problem for some arbitrary circuit. If all primitive sets are of size one, we can condense each chain of blocks to a single number denoting where the non-air block is, and similarly we can represent each stair in the staircase as a single number. \textbf{b)} The condensed notation for the same problem. We sometimes refer to this notation as a "ladder" in this work. Numbers summing to same value in the ladder notation corresponds to blocks occupying the same column in the staircase notation. For problems where the ladder notation is applicable (primitive sets of size one), both the ladder and staircase notation are equivalent, however the ladder notation is considerably more concise.
}
    \label{fig:ladder_staicase_translation}
\end{figure} 

\subsection{Ancillary Qubit Re-indexing is NP-hard\label{sec:NP_proof}}
We prove here that the assignment problem we devised heuristics for in Section~\ref{sec:heuristics} and Section~\ref{sec:shor} is in fact NP-hard. We do so by first noting that if we could solve the general problem of naive-syndrome re-indexing defined in Section~\ref{sec:heuristics}, then we could also solve the special case of Shor-syndrome re-indexing, as Shor-syndrome re-indexing is simply the case where all primitive sets are of size 1. Therefore the complexity of the general problem is at least as hard as the Shor-syndrome re-indexing problem. With this in mind, we prove that Shor re-indexing is NP-hard. First let's recall the problem definition from Sec~\ref{sec:shor}: \\
 
\textbf{Shor Re-indexing Problem Formulation}: Given a multi-set of $s$ integers $\mathcal{S} = \{p_1,p_2,...,p_s\}$, find a bijection $\pi: [s] \rightarrow [s]$, such that $|\mathcal{G}|_0 $ is minimized, which can be written as follows: 
\begin{equation}
    \min_\pi |\mathcal{G}|_0  = \min_\pi \Bigl| \bigcup_{i=1}^s \{p_i+ \pi(i) \}\Bigr|
\end{equation}

We would like to prove that given an $\mathcal{S}$, with its most common element occurring $t$ times, determining whether or not there exists a bijection $\pi$, such that $\Bigl| \bigcup_i \{p_i+ \pi(i) \}\Bigr| =t+1$, to be NP-hard. 

We prove this by reduction from 3-partition. In the 3-partition problem, we are given $3m$ positive integers as a multiset $A=\{a_1,...,a_{3m}\}$. The problem is to see if we can partition these into $m$ triples (multisets of size 3), such that each set sums to the same value $T$, and we may assume that $T/4 < a_i < T/2$~\cite{3_partition}. This is an NP-complete problem so if we can show that any algorithm that could solve our problem generally would also be able to solve arbitrary instances of 3-partition under a polynomial time mapping to our problem, then we know that our problem must be at least as hard as 3-partition, and therefore NP-hard. 

For this reduction, we take an arbitrary 3 partition problem and map it to our Shor re-indexing problem  by mapping the elements of $A$ to a multiset $\mathcal{S}$, where $\mathcal{S}$ is the input to our problem as defined above. This $\mathcal{S}$ can then be decomposed into no fewer than $3m+1$ sets, each corresponding to a chain $c_i$, where $i\in \{0,1,...,3m\}$. Rather than specifying each element in $\mathcal{S}$, we instead define each $c_i$ below, noting that $\mathcal{S}= \sum_{i=0}^{3m}c_i$ (addition of multisets here is defined as combining them, e.g.\ $\{1,2\}+\{2,3\}=\{1,2,2,3\}$.) Later we prove that with how our $\mathcal{S}$ is defined, no other set of chains would be feasible in obtaining a solution to the reindexing problem, thus excluding the possibility of a solution to Shor re-indexing existing that does not also map to a valid solution of the given 3-partition problem. We define $c_0,c_1,...,c_{3m}$ as follows:

\begin{itemize}
    \item First, let:
    \begin{align}
        a^* = \max_{A} - \min_{A} < \frac{T}{2} - \frac{T}{4} \leq \frac{T}{4}
    \end{align}

    We use this to define $c_1, ..., c_{3m}$, noting that all of these chains each correspond to an original element in $A$:
    \begin{align}
        c_i = \{0,1,..., a^* + a_i-1\}, \forall i \in \{1,...,3m\}
    \end{align}
    
    Each chain is length $a+a^*$, corresponding to each $a \in A$, and each chain has at least the first $a^*$ elements in common. 
   
    \item $c_0$ is defined as follows:
    \begin{align}
        F(i,j) = i(T+3a^*)+j+a^*(i-1)+(|S|+T) \\
        c_0  = \left\{ (|S|+T), F(i,j) | i\in \{1,...,m\}, j \in \{1,...,a^*\} \right\}
    \end{align}
\end{itemize}

Later we prove, as they are invoked, that this definition of $c_i,$ where $i \in \{0,1,...,3m\}$ guarantees following lemmas:
\begin{enumerate}
    \item $c_0$ contains $m$ gaps, each of size $T+3a^*$
    \item Between each gap in $c_0$, there exists a run of $a^*$ elements
    \item It is not possible for any $e \in c_0$ to participate in another chain.
    \item If we were to construct an alternate collection of $3m$ sets (chains) from elements in the multi-set $\sum_{i=1}^{3m}c_i$, then that alternate collection would require at least one chain with a gap in it. 
\end{enumerate}

Finally we notice that $0$ occurs $3m$ times, so this construction needs at least $3m$ shuttles. Furthermore, due to 3) above, our construction will require a minimum of $3m+1$ shuttles, and our problem is to determine whether or not a bijection $\pi: [s] \rightarrow [s]$, recall $s = |\mathcal{S}| = \sum_{i=0}^{3m}|c_i|$, exists such that:
\begin{align}
    \Bigl| \bigcup_{p \in \mathcal{S}} \{p+ \pi(i) \}\Bigr| = 3m+1
\end{align}

To structure our proof, we first prove that if a solution exists in the 3-partition instance, then a solution of exactly $3m+1$ shuttles must exist for our problem. We prove afterwards that if a solution to our problem exists, then a solution to the 3-partition problem must also exist. \\


\textbf{Direction 1: 3-Partition $\rightarrow$ Shor Re-indexing}:
If a solution to the 3-Partition instance exists, then we know that $A$ can be partitioned into $m$ triples such that each triple sums to $T$. In the Shor re-indexing problem, we can place $c_0$ in the only allowable position (placement refers to the "ladder" picture as in Fig.~\ref{fig:steane_chains}). Invoking Lemma 1, we know that there will be $m$ gaps of size $T+3a^*$ to place the remaining elements after placing $c_0$. For each of the $m$ triples given by the 3-partition solution, we consider the corresponding triple of chains, that is, let us consider an arbitrary triple $(c_i, c_j, c_k)$ such that we know $a_i+a_j+a_k = T$. Chain $c_i$ has a length equal to $a_i+a^*$, therefore placing these three chains sequentially would occupy the following amount of space:

\begin{align*}
    (a_i+a^*)+(a_j+a^*)+(a_k+a^*)= a_i+a_j+a_k+3a^*  \\
    \quad = T+ 3a^* 
\end{align*}

Therefore each triple of chains would be packable into the gaps that occur within $c_0$. Therefore, all chains can be packed into $s$ positions exactly, implying the existence of a $3m+1$ shuttle solution.

\begin{lemma}
$c_0$ contains exactly $m$ gaps, each of size $T+3a^*$
\end{lemma}

\begin{proof}
In the construction of $c_0$, we notice that the difference of the largest and smallest element is
\begin{align*}
    F(m,a^*) - (|\mathcal{S}|+T) = m(T+3a^*)+a^*+a^*(m-1) \\
    = m(T+4a^*)
\end{align*}
Therefore $m(T+4a^*) -1$ is the number of spaces between the first and last element, where the $-1$ is due to the number of spaces between, for example, $2$ and $0$ being $1$ in this context. As provided by Lemma 2, in those spaces, we have $(m-1)$ ridges of size $a^*$, one ridge of size $a^*-1$. So there are a total of $(m-1)a^*+(a^*-1)$ elements in this spaces, meaning that the total size of all gaps in $c_0$ is given equal to :

\begin{align*}
    F(m,a^*) - F(1,1 ) - 1 - ((m-1)a^*+(a^*-1)) \\
    = m(T+4a^*) - 1 - ((m-1)a^*+(a^*-1)) \\
    = m(T+4a^*) - 1 - (ma^*-a^*+a^*-1) \\
    = m(T+4a^*) - 1 - (ma^*-1) \\
    = m(T+4a^*) - ma^* \\
    = m(T+3a^*)
\end{align*}

Thus, there is a total of exactly $m(T+3a^*)$ gaps in $c_0$ as defined. Now to prove this corresponds to exactly $m$ gaps, where each is size $T+3a^*$. We show this by first noticing that aside from the first element, the gaps in $c_0$ occur on between elements $F(n+1, 1)$ and $F(n, a^*)$, for some $n \in \{1,...,m-1\}$. The first gap is size:
\begin{align*}
F(1,1) - |S|+T  -1 \\
=(T+3a^*)+1 + (|S|+T) - (|S|+T) - 1\\
= (T+3a^*)
\end{align*}
Then for all other gaps:
\begin{align*}
    F(n+1,1)-F(n,a^*) -1 \\
    = (n+1)(T+3a^*)+1+a^*n + (|S|+T) \\
    - (n(T+3a^*)+a^*+a^*(n-1)+ (|S|+T)) - 1\\
    = T+3a^* + a^*n - (a^*+a^*(n-1))\\
    = T+3a^*
\end{align*}

Therefore the lemma holds, and we have shown that $c_0$ contains exactly $m$ gaps of each size $T+3a^*$.
\end{proof}

\begin{lemma}
    Between each gap in $c_0$, there exists a run of $a^*$ elements
    \begin{proof}
        We first notice that if written in ascending order, the elements in $c_0$ are $|S|+T, F(1,1), F(1,2),...,F(1, a^*),F(2, 1),...,F(m,a^*)$. Gaps occur between the first and second element and also between $F(n+1,1)$ and $F(n,a^*)$, as discussed in Lemma 1. Therefore between $F(n,1)$ and $F(n,a^*)$ are $a^* - 2$ elements, and a total of $a^*$ elements following the gap defined by the difference of $F(n,1)$ and $F(n-1,a^*)$. This is true for all $n \in \{2,3,...,m\}$, as well as between the first and second element, that is between $F(1,1)$ and $|S|+T$.

        Therefore each gap in $c_0$ is followed by $a^*$ elements.
    \end{proof}
\end{lemma}

\textbf{Direction 2: Shor Re-indexing $\rightarrow$ 3-Partition}:
Now we prove that any $3m+1$ shuttle solution to Shor re-indexing would also be a solution to the 3-partition problem. We prove this by proving that our construction could only ever provide a $3m+1$ shuttle solution if our definitions of $c_0,c_1,...,c_{3m}$ are followed exactly, which would then mean that each chain, aside from $c_0$, represents an element from $A$ encoded in unary, and as discussed in Direction 1, the packing of those chains corresponds to triples from $A$ each summing to $T$. In other words, we designed our reduction so that technically speaking, the idea of packing chains is not enforced, however by design, the only possible $3m+1$ shuttle solution will also be the solution to packing our chains as defined, and no other collection of chains would be feasible. To show that no other collection if feasible, we first prove the remaining two lemmas.

\begin{lemma}
    It is not possible for any $e \in c_0$ to participate in another chain.
    \begin{proof}
        We observe that $a^*< \frac{T}{4}$, and that $a < \frac{T}{2}$ for all $a \in A$. Therefore $\max_{e \in \sum_{i=1}^{3m} c_i}e < \frac{T}{4} + \frac{T}{2}$. The smallest value in $c_0$ is $|S|+T$, so if the largest element in $\sum_{i=1}^{3m} c_i$ and the smallest element in $c_0$ were in the same chain, a gap of size at least $|S|+T/4-1$ would be created. Because we are looking for a bijection $\pi$ that maps from $1,...,|S|$ to $1,...,|S|$, the longest possible chain (including its gaps) would be limited to length $|S|$, and the largest gap possible would be $|S|-2$. Therefore, because $|S|+T/4-1 > |S|-2$, as $T/4 > 0$, it would not be possible for any element of $c_o$ to participate in one of the other chains.
    \end{proof}
\end{lemma}

\begin{lemma}
    If we were to construct an alternate collection of $3m$ sets (chains) from elements in the multi-set $\sum_{i=1}^{3m}c_i$, then that alternate collection would require at least one chain with a gap in it. (A chain with a gap in it refers to a nonconsecutive set of integers)
    \begin{proof}
        First we notice that all chains $c_1,...,c_{3m}$ have the first $a^*$ elements in common by definition. Because chains are defined as sets and not multisets, an alternate set of $3m$ chains would also need to start with same first $a^*$ elements in common. Furthermore, we notice that because all chains have the first element in common, for an alternate set $c_1',...,c_{3m}'$ to be distinct from $c_1,...,c_{3m}$, it would need to contain some chains with gaps in them. This is because chains without gaps correspond to values in $A$ plus a constant $a^*$, encoded in unary. An alternate set that also didn't have any gaps would correspond to the same multi-set of numbers in $A$ encoded in unary, and therefore $\{c_1,...,c_{3m}\}$ and  $\{c_1',...,c_{3m}'\}$ would be equivalent.
        
        To illustrate this further, consider our input multiset $\sum_{i=1}^{3m}c_i$, and consider creating $3m$ empty sets that will be our chains. If we take all the 0's from our input and distribute them among these chains, and then 1's and do the same, and so on, eventually we will arrive at a number $n_1$ with less than $3m$ copies. The chains that don't receive a copy of $n_1$ will no longer need to be considered if we are interested in generating sets of chains without gaps (as the lack of $n_1$ means we can't add further elements without inducing a gap). If we repeat this process until all values in our input have been assigned to chains, then we naturally will be left with a collection of chains that have lengths equal to the values in $A$ plus the constant $a^*$. Is should be straightforward to see there is only one collection of gap-free chains possible in this problem, where a collection has no ordering to it.
        
   Therefore we can say that any set of $3m$ chains partitionable from $\sum_{i=1}^{3m} c_i$ yet distinct from $\{c_1,...,c_{3m}\}$  would need to include some gaps in them. 
    \end{proof}
\end{lemma}

We use the above lemmas to prove that no other collection of chains would result in a $3m+1$ solution. For the sake of contradiction, let's assume that such an alternate set exists. By Lemma 3, we know that $c_0$ must remain unchanged, as $c_0$ constitutes $1$ shuttle, and the other $3m$ shuttles must come from the other elements, as $0$ occurs $3m$ times. Therefore this alternate set of chains would be $\{c_0,c_1',...,c_{3m}'\}$, and a gap must occur at least once in a chain in $\{c_1',...,c_{3m}'\}$ by Lemma 4. By Lemma 1, we know there are $m$ gaps of size $T+3a^*$, and thus to have a $3m+1$ solution, we must be able to pack $\{c_1',...,c_{3m}'\}$ into those $m$ gaps of size $T+3a^*$. Therefore the gaps introduced by using an alternate set of chains must be filled by other chains. As all gaps are by definition of size less than $a^*$, the ridges (Lemma 2) within $c_0$ cannot be used to fill them, which then leaves the only possibility to be that within each of the $m$ gaps in $c_0$, other chains would fill gaps. However, the structure of any possible packing of chains from $\{c_1',...,c_{3m}'\}$ into a triple $(c_i,c_j,c_k)$ would be a run of at least $a^*$ elements , followed potentially by a region of gaps with possible some elements in between (all of those gaps are size less than $a^*$), followed again by a run of at least $a^*$ elements, another region of gaps, and so on. This corresponds to each chain beginning with a consecutive sequence of at least $a^*$ elements. Because of this, it is not possible for the beginning of say chain $c_j$ to fill the gaps present in chain $c_i$. Therefore this alternate set of chains would not give way to a valid $3m+1$ solution, as it would be impossible to pack, and therefore there will not exist a bijection that satisfies Equation (9). Therefore by contradiction, the only possible $3m+1$ solution is to use the initially defined sets $\{c_0,c_1,...,c_{3m}\}$.

Therefore, if we have a $3m+1$ solution to Shor re-indexing, then it follows that we have successfully packed the elements of $A$ plus some constant $a^*$, encoded in unary, into $m$ gaps, each of size $T+3a^*$. From this point, we can recover a solution to the arbitray 3-partition problem instance.

Therefore Shor re-indexing is at least as hard as 3-partition. Therefore Shor re-indexing is NP-hard.

\subsection{Preparation of Cat States\label{sec:cat_prep}}
Here, we present a possible simple algorithm for reordering a linear chain of qubits, provided we have access to another blank row of qubits, which can then be applied to generating the interwoven cat states that go along with compiled Shor-syndrome circuits, as we assume that nearest neighbor cat states can be generated.  

\begin{figure}
    \centering
    \includegraphics[width=\linewidth]{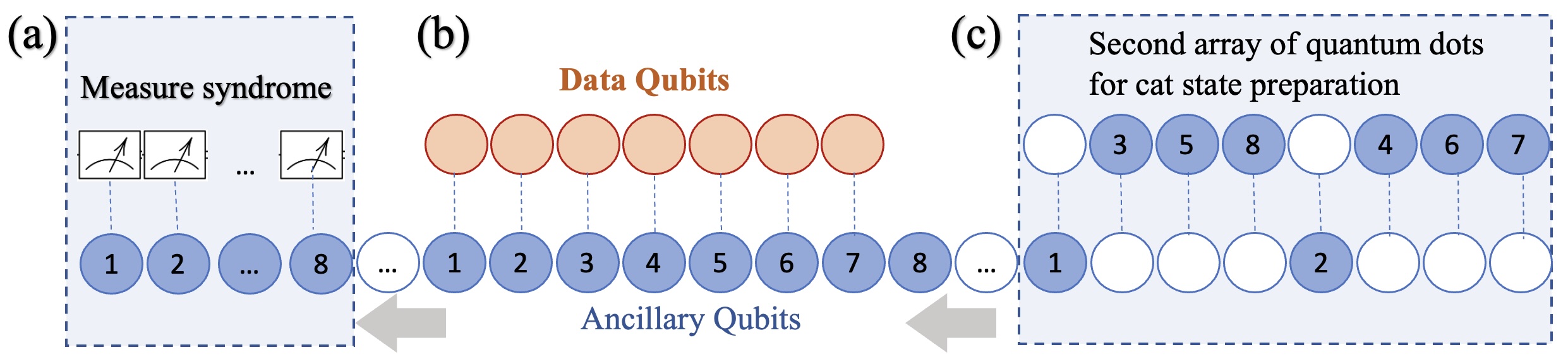}
    \caption{\textbf{Illustration of the full hardware running a Shor syndrome circuit}. a) After performing syndrome extraction by shuttling our qubits into the proper configurations and applying the appropriate two-qubit gates, the ancilla must then be shuttled to an area where quantum measurement can be carried out. Here it is represented as a continuation of the 1-D architecture, however this is mostly illustrative. b) After generating the required ancilla in their proper interwoven cat states, the ancilla are then shuttled into the data qubit area where we perform syndrome extraction, shuttling to different values of $\delta$, as described in the main text. c) In this area, we propose using 2 rows of quantum dots to generate the interwoven cat states required by how we compiled the syndrome extraction circuit. 
}
    \label{fig:full_2xn}
\end{figure}

\begin{figure}
    \centering
    \includegraphics[width=\linewidth]{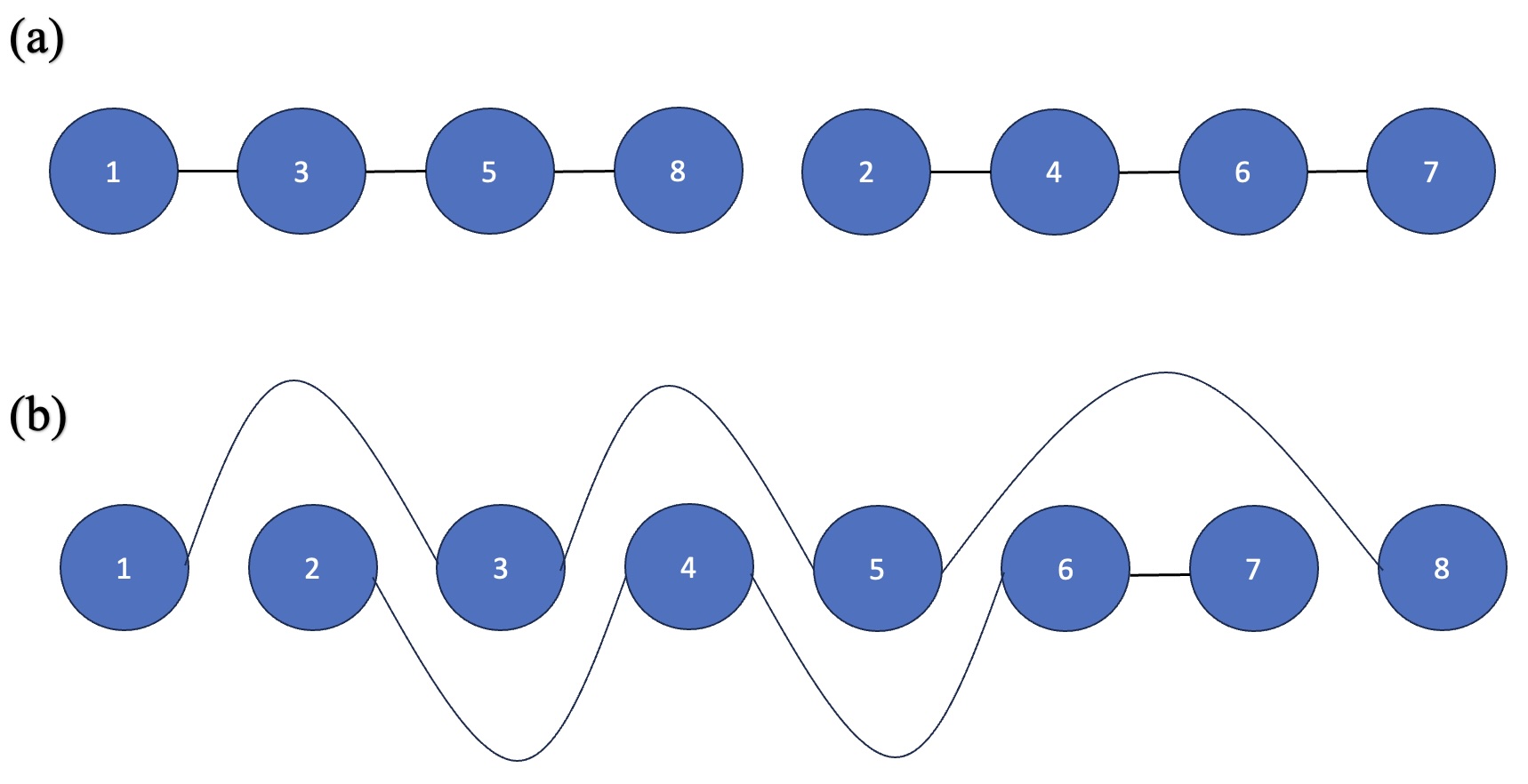}
    \caption{\textbf{Illustration of the cat state preparation problem}. a) An initial collection of cat states, generated utilizing nearest neighbor connectivity. 
    b) The order of cat states necessitated by the compilation of a Shor syndrome circuits. As this is the desired order, we enumerates these consecutively. "Un-weaving" this gives us the numbers present in a), and now all we need to do is use empty quantum dots to sort our initial sequence of cat states.
}
    \label{fig:cat_prep}
\end{figure}

Referring to Fig.~\ref{fig:cat_prep}a, we start with an array of adjacent cat states. After applying the bijection $\pi$ found by compilation to the ancillary qubit indices, we will often need to reorder this initial cat state such that the each stabilizer check is still being interacted with the correct state. Fig.~\ref{fig:cat_prep}b is some example arbitrary ordering we might need to obtain. To achieve this we generate the cat states in an adjacent array of quantum dots, as shown in Fig.~\ref{fig:full_2xn}c. We now can read the labels of the qubits from left to right, shuttling to the main array the largest consecutive sequence that starts at 1 (the first step of this is depicted in Fig.~\ref{fig:full_2xn}c.) Next we can shuttle these qubits to the left and out of the way of the remaining qubits in the top row, so that we can repeat this process until all ancillary qubits on present on the bottom row of qubits, in the proper order. Each time we shuttle down qubits, we start the sequence at one more than the previous sequences last element. For example, again in Fig.~\ref{fig:full_2xn}c, the second step would shuttle down qubits 3 and 4, the third step would shuttle down 5, 6, 7, and the final would bring down qubit 8.

\subsection{Simulation Results\label{sec:sim_results}}
Through simulation, we show that our methods of circuit compilation can be used to reduce the threshold of various error correcting codes. We model Pauli noise with highly efficient Pauli-Frame simulation techniques~\cite{stim,knill_pf} implemented in Julia~\cite{quantumclifford}. We simulate at least 20,000 samples for each set of noise parameter values (each data point in our plots) in order to calculate the logical error rate for a single round of error correction. The simulated circuits contain the following steps:

\begin{figure}
    \centering
    \includegraphics[width=\linewidth]{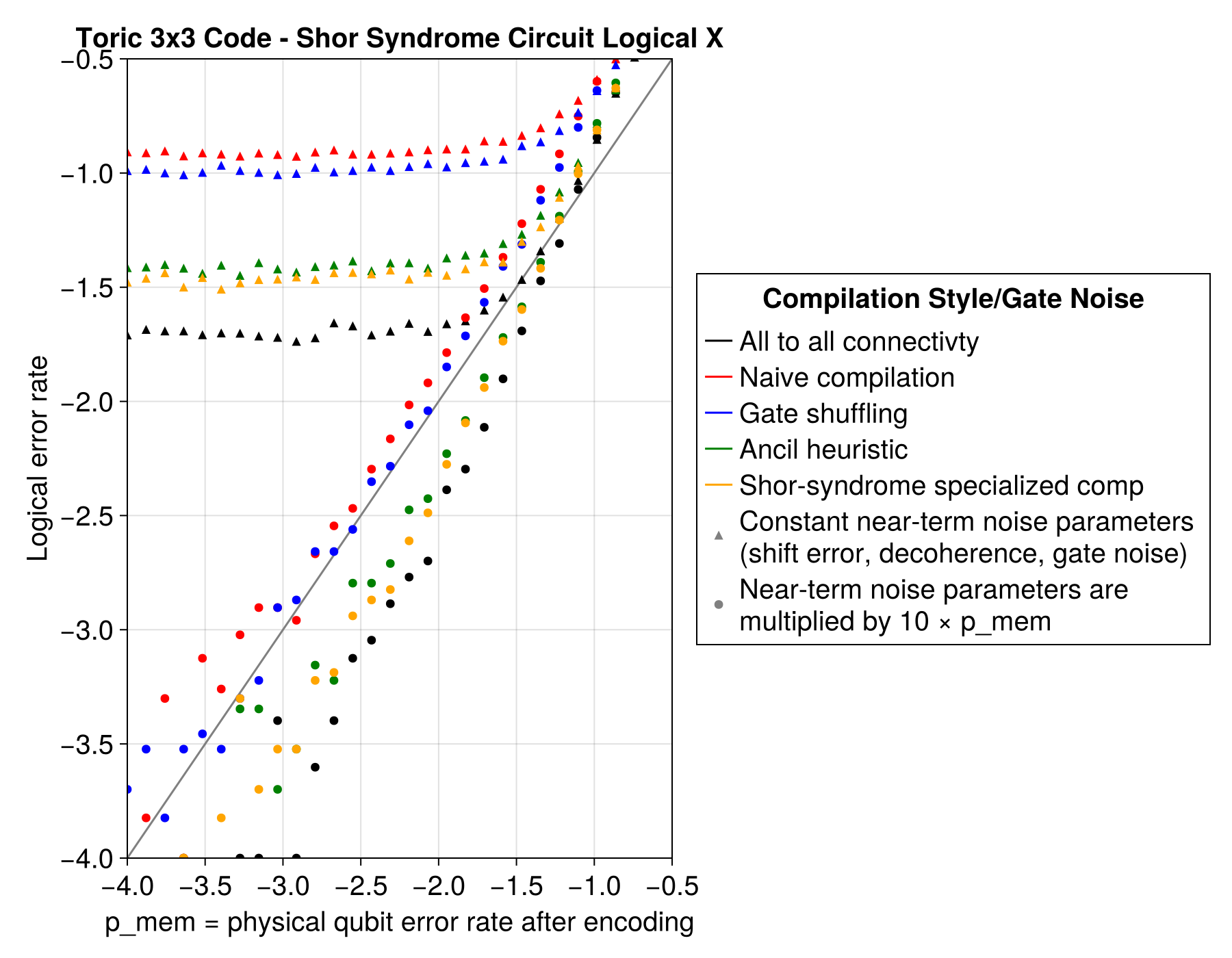}
    \caption{\textbf{Performance of Toric 3x3 code's Shor syndrome circuit on $2 \times n$ architecture}. The $y$-axis is $\log_{10}$ of the logical error rate, and the $x$-axis is $\log_{10}(p_\text{mem})$. The different colors correspond to different levels of compilation, noting that "Naive compilation" corresponds to running canonical syndrome extraction circuits without any specialization to this hardware. "All to all" corresponds to setting $p_\text{wait}$ and $p_\text{shuttle}$ to 0. This is only the plot of the logical $X$ error rate, although the plot for the logical $Z$ error rate looks quite similar. See Fig.~\ref{fig:toric_plots} for the logical $Z$ error rate, as well as the plots of this code's performance on a naive syndrome extraction circuit.
}
    \label{fig:toric3x3_plot}
\end{figure}
\begin{itemize}
    \item \textbf{Encoding} For the encoding step, one can use either an encoding circuit or measure the stabilizers to project onto a stabilized space, tracking initial corrections as needed (in Pauli frame simulation, these corrections can be ignored). As the latter is more resilient to error, that is how we carry out encoding in our simulations, although the same result can be achieved by simulating an error-free encoding circuit. 
    \item \textbf{Memory noise} After encoding, we apply a round of Pauli noise on each data qubit with independent probability $p_{\text{mem}}$. This is the typical place to insert errors when evaluating a code or decoder in non-fault-tolerant simulation. \\
    \item \textbf{Syndrome measurement} Syndrome measurement will consist of a few pieces. First is the syndrome circuit itself, which is the output of our compilation algorithm (our algorithm transforms a syndrome circuit into an equivalent syndrome circuit). Then we initialize a variable $\delta_{\text{current}}=0$ to track the current state of the $2 \times n$ hardware. If the next gate $\mathcal{G}_{\text{next}}$ to be run in the syndrome circuit has $\delta(\mathcal{G}_{\text{next}}) \neq \delta_{\text{current}}$, then we simulate the noise from a shuttle by applying Pauli Z errors independently to each of the idling data qubits with probability $p_\text{wait}$, and likewise we apply Pauli noise to each ancilla qubit with probability $p_\text{shuttle}$. Afterwords, we set $\delta_\text{current} = \delta(\mathcal{G}_{\text{next}})$.  Furthermore, we apply Pauli noise on both qubits after each two-qubit gate with probability $p_{\text{gate}}$.
     \item \textbf{Logical operator measurement} In order to have something to compare our after-decoding results to, we measure the logical X \textit{or} Z operators at the end of each syndrome circuit, as if it were an additional parity check. More specifically, we cannot measure both operators in the same circuit, so we need to run two simulations in order to obtain both the logical X and Z error rates.
     \item \textbf{Decoding} 
     We now run our decoder of choice on the syndromes, and see if the predicted error would cause the same logical error reported in the logical X/Z syndrome. If we predicted the logical error correctly, then that counts as successfully decoded. The logical error is then defined as the ratio of successfully decoded samples to total samples. A lookup table decoder is used for small ($d=3$) codes such as Steane's [[7,1,3]] code~\cite{steane7}, and MWPM~\cite{pymatching} is used for toric codes. 
     \end{itemize}

In our plots, we vary $p_{\text{mem}}$ along the horizontal axis and use constant values for other types of error. Specifically, we use a value of $p_\text{shuttle} = 0.0001$ from~\cite{zwerver23shuttling}, a potentially near-term achievable value of $p_{\text{gate}}=0.9995$, and $p_\text{wait}= 1-e^{\frac{-(2\mu\text{s}+12.5\mu\text{s})}{28\text{ms}}}$. The shuttle time was said to be $2\mu$s ramp and the fastest time between shuttles was stated to be $12.5\mu$s in~\cite{zwerver23shuttling}. $28$ms was the T2 time reported in~\cite{Veldhorst_2014} (recently great effort has gone into improving to these times~\cite{Wang2024}). These near-term parameters are represented by the triangles in Fig.~\ref{fig:toric3x3_plot}. Notice that these curves have error floors due to constant error terms that do not depend on $p_\text{mem}$. However, the curves plotted by circles correspond to scaling $p_\text{wait}$, $p_\text{shuttle}$, and $p_\text{gate}$ each by $10 \times p_\text{mem}$, creating curves with pseudo-thresholds, as one may be used to. 

Fig.~\ref{fig:toric_plots} presents the compilation performance of both naive and Shor syndrome circuits for the 3x3 toric code. These plots are similar to the style of Fig.~\ref{fig:toric3x3_plot}, except we now have both the logical X and Z error rates corresponding to the top and bottom rows of the figure respectively. The left column corresponds to the performance on naive syndrome circuits, while the right column corresponds to performance on fault-tolerant Shor-syndrome circuits.

\newpage
\pagebreak

\begin{figure}[tb]
    \centering
    \includegraphics[width=.9\paperwidth]{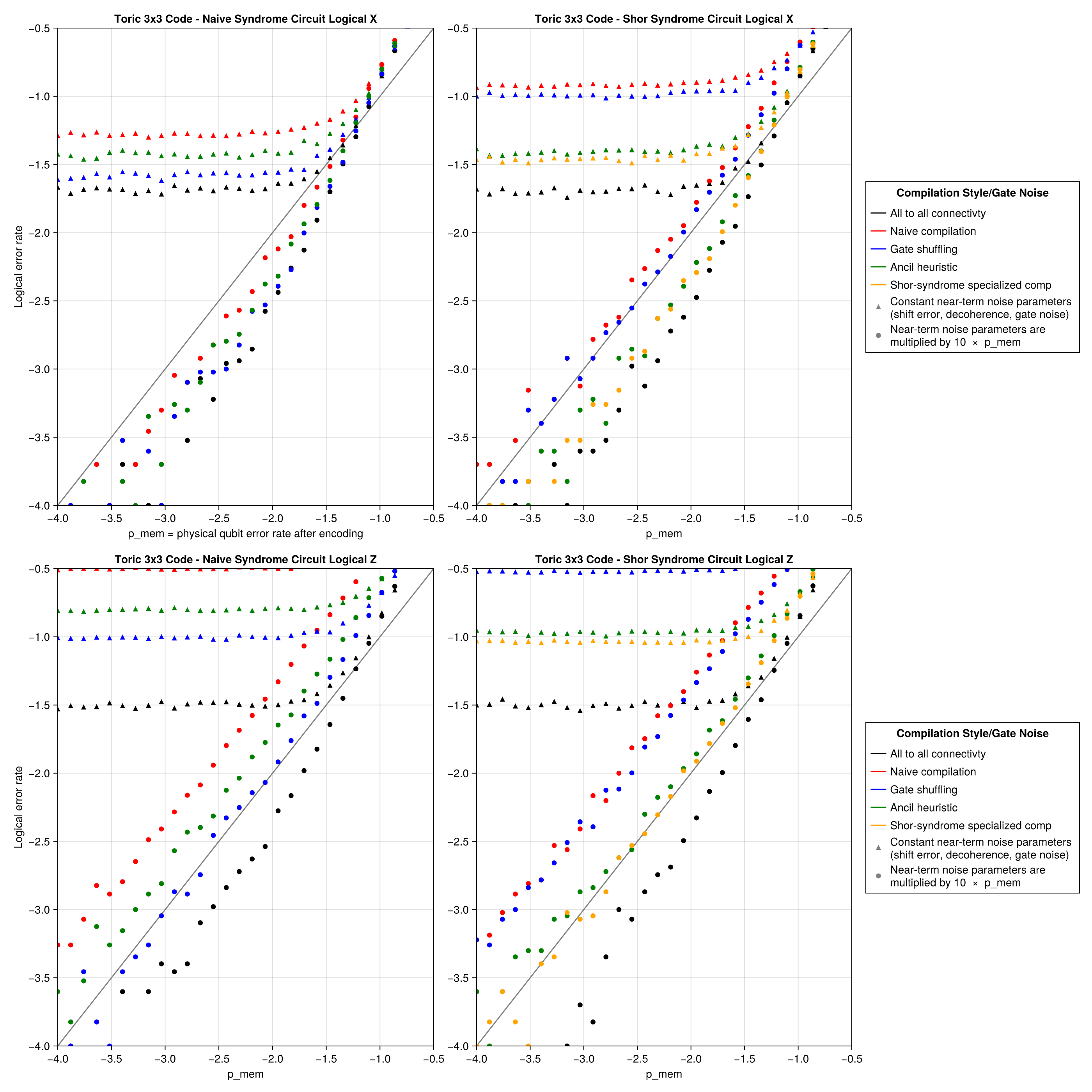}
    \caption{\textbf{Compilation performance of the 3x3 Toric code}. The toric code's naive syndrome measuring circuit is an example of where the heuristics break down a little. They still perform better than doing nothing, but as shown here, and in Table~\ref{tab:naive_XZ}, gate shuffling alone finds a decent solution of 12 shuttles regardless of the lattice size, whereas the ancillary qubit re-indexing heuristics find a number of shuttles that grows with lattice size.
}
    \label{fig:toric_plots}
\end{figure}

\end{document}